\documentclass[a4paper,11pt]{article}

\usepackage{jheppub} 

\usepackage[T1]{fontenc} 

\usepackage{graphicx} 
\usepackage{amsthm}

\usepackage{bbm}

\newtheorem{thm}{Theorem}[section]

\newtheorem{defn}[thm]{Definition}

\newcommand{\abs}[1]{\left| #1 \right|} 
\newcommand{\ket}[1]{| #1 \rangle } 
\newcommand{\bra}[1]{\langle  #1 |} 
\newcommand{\braket}[2]{\langle  #1 \vphantom{#2} | #2 \vphantom{#1} \rangle } 
\newcommand{\ketbra}[2]{|  #1 \vphantom{#2} \rangle \langle #2 \vphantom{#1} | } 
\newcommand{\diracprod}[3]{\left\langle  #1 \vphantom{#2#3} \right|#2 \left| #3 \vphantom{#1#2} \right\rangle } 
\let\baraccent=\= 
\renewcommand{\=}[1]{\stackrel{#1}{=}} 

\newcommand{\im}{\text{Im}}

\def\A{\mathcal{A}}
\def\B{\mathcal{B}}
\def\F{\mathcal{F}}

\def\O{\mathcal{O}}
\def\P{\mathcal{P}}

\def\L{\mathcal{L}}

\def\H{\mathcal{H}}

\def\C{\mathbb{C}}
\def\Rbb{\mathbb{R}}
\def\Tr{\text{Tr}}

\def\Srel{S_{\text{rel}}}
\def\id{\mathbbm{1}}

\title{Semiclassical algebraic reconstruction for type III algebras}

\author[a]{Haocheng Zhong}

\affiliation[a]{Shing-Tung Yau Center and School of Physics, Southeast University, Nanjing 210096, China}

\emailAdd{zhonghaocheng@seu.edu.cn}

\abstract{In this work, we address the unresolved type III cases of the algebraic reconstruction theorem by integrating crossed product algebras and semiclassical approximations. We first derive that the relative entropy in crossed product algebras factorizes into contributions from the original algebra and observer wavefunctions. By constructing ``holographic'' crossed product algebras for ``bulk'' and ``boundary'' type III factors, we extend the algebraic reconstruction theorem to include the algebraic Ryu-Takayanagi (RT) formula semiclassically, which provides a complete algebraic description of the reconstruction theorem, as an intrinsic framework for the algebraic version of bulk-boundary correspondences in holographic duality.}

\begin{document} 
\maketitle
\flushbottom

\section{Introduction}

In recent years, the application of von Neumann algebra in quantum gravity has made significant progress, particularly in understanding the nature of entanglement structure of the spacetime \cite{Witten:2021unn,Chandrasekaran:2022cip,Jensen:2023yxy,Kudler-Flam:2023hkl,Colafranceschi:2023moh,Faulkner:2024gst,Kudler-Flam:2023qfl,Kudler-Flam:2024psh,AliAhmad:2023etg,Leutheusser:2021qhd,Leutheusser:2021frk,Leutheusser:2022bgi,Engelhardt:2023xer}. Essentially, the divergent property of the entanglement entropy in a quantum field theory is due to the fact that the theory is described by a type III$_1$ algebra \cite{Yngvason:2004uh,Fredenhagen:1984dc,Araki:1963klf,Araki:1964lyc}, which does not admit a well-defined notion of von Neumann entropy \cite{Witten:2018zxz}. Nevertheless, it has been demonstrated that in a quantum theory coupled to gravity, one is able to perform an algebraic transition from a type III$_1$ algebra to a type II$_\infty$ algebra, which does admit a well-defined von Neumann entropy, via the crossed product construction \cite{Witten:2021unn,Chandrasekaran:2022eqq,Chandrasekaran:2022cip,Penington:2023dql,Kolchmeyer:2023gwa}. Such construction is extensively used in black hole physics where under certain geometric constraints, the generalized entropy and the generalized second law are formalized in the algebraic language \cite{Jensen:2023yxy,Kudler-Flam:2023qfl,Faulkner:2024gst,Klinger:2026tws}. See also \cite{Lin:2022rbf,Witten:2023qsv,DeVuyst:2024khu,Sorce:2024zme,Fewster:2024pur,DeVuyst:2024fxc,Kudler-Flam:2024psh,Xu:2024hoc,AliAhmad:2024wja,Bahiru:2022oas,Klinger:2023auu,Gomez:2023wrq,AliAhmad:2024eun,vanderHeijden:2024tdk,AliAhmad:2024saq,AliAhmad:2025oli,Klinger:2025tvg,Klinger:2026tws,Klinger:2026kqj,Aguilar-Gutierrez:2023odp,Espindola:2026ekv} for a incomplete list of relevant developments regarding algebraic approach in gravity.

In the context of the AdS/CFT correspondence \cite{Maldacena:1997re,Witten:1998qj,Gubser:1998bc}, which conjectures an equivalence between $d-$dimensional quantum gravity on anti-de Sitter spacetime and $(d-1)-$dimensional conformal field theory, the fact that a type III algebra does not admit a well-defined von Neumann entropy implies the famous Ryu-Takayanagi (RT) formula \cite{Ryu:2006bv,Hubeny:2007xt,Casini_2011,Lewkowycz:2013nqa,Nishioka:2018khk,Faulkner:2013ana,Engelhardt:2014gca} is in fact mathematically ill-defined for the lack of formal definition of entropies both in the bulk and on the boundary. The RT formula is stated as 
\begin{equation}\label{eq:RT}
	S(\rho_A)=\mathcal{L}_A+S(\tilde{\rho}_a),
\end{equation}
provided a density operator $\rho_{A}$ on the boundary subregion $A$ and the dual density operator $\tilde{\rho}_a$ in the entanglement wedge $a$, and $\mathcal{L}_A$ is the area of the Ryu-Takayanagi surface of $A$ over $4G$. The RT formula is so essential that in the holographic framework, the formula itself is equivalent to the following three statements according to a series of works \cite{Almheiri_2015,Dong:2016eik,Harlow:2016vwg}\footnote{See also appendix A of \cite{Zhong:2024fmn} for a summary proof, and \cite{Faulkner:2020hzi,Akers:2021fut,Gesteau:2023hbq} for an incomplete list of related developments.}, which we call the reconstruction theorem:
\begin{itemize}
	\item Entanglement wedge reconstruction (or subregion duality) \cite{Czech:2012bh,Hamilton:2006az,Morrison:2014jha,Bousso:2012sj,Bousso:2012mh,Hubeny:2012wa,Wall:2012uf,Headrick:2014cta,Jafferis:2015del,Dong:2016eik}, establishes a correspondence between a boundary subregion and its bulk entanglement wedge, where any bulk operator contained within the entanglement wedge can be reconstructed from information residing on the designated boundary subregion;
	
	\item The Jafferis--Lewkowycz--Maldacena--Suh (JLMS) formula \cite{Jafferis:2015del} establishes a duality between boundary and bulk descriptions of quantum information. Consider two density operators, $ \rho_{A} $ and $ \sigma_{A} $, defined on a boundary subregion $ A $, whose associated entanglement wedge is denoted by $ a $. The relation states that the quantum relative entropy evaluated on the boundary coincides with its bulk counterpart:
	\begin{equation}
		\Srel(\rho_A|\sigma_A)=\Srel(\tilde{\rho}_a|\tilde{\sigma}_a),
	\end{equation}
	where $ \Srel $ denotes the quantum relative entropy, and $ \tilde{\rho}_{a} $ and $ \tilde{\sigma}_{a} $ are the bulk dual density operators corresponding to $ \rho_{A} $ and $ \sigma_{A} $, respectively;
	
	\item Radial commutativity \cite{Polchinski:1999yd,
		Almheiri_2015,Harlow:2018fse}, which states that bulk operators at a fixed bulk time-slice commute with all boundary operators localized at the boundary of that time-slice.

\end{itemize}

The equivalence of the above three statements are generalized into algebraic descriptions in \cite{Kamal:2019skn,Kang:2018xqy,Kang:2019dfi,Xu:2024xbz,Crann:2024gkv} (we call it the algebraic reconstruction theorem), but the inclusion of the RT formula is problematic due to the subtleties discussed previously. However, for the cases of type I/II algebras, such inclusion is plausible as proved in \cite{Xu:2024xbz}, leaving only the type III cases, where the algebras of a generic quantum field theory belong to, unresolved. On the other hand, applying semiclassical approximation is a common tool in studying quantum gravity. Especially in the context of holographic framework, the semiclassical bulk geometry should be consistent with the RT formula and the JLMS formula \cite{Lashkari:2013koa,Faulkner:2013ica,Swingle:2014uza,Faulkner:2017tkh,Lewkowycz:2018sgn,Jafferis:2015del}, such that it is intriguing to investigate how the semiclassical approximation can be adopted into the algebraic reconstruction theorem to resolve the type III issues. In this work, motivated by the success of the crossed product construction to black hole physics, we apply the construction to the holographic framework and impose the semiclassical approximation introduced in \cite{Chandrasekaran:2022cip,Chandrasekaran:2022eqq,Jensen:2023yxy}. The main results of the work are listed in the following:
\begin{itemize}
	\item In the semiclassical description, we prove that the relative entropy in the crossed product algebra factorizes into a sum of two terms, one of which is the relative entropy of the original algebra while the other corresponds to the relative entropy between two observers' wavefunctions:
	\begin{equation}
		\begin{aligned}
			\text{Semiclassically,}\quad\Srel(\Psi_f|\Phi_g;\widehat{\A}_{\Omega})\approx\Srel(\Psi|\Phi;\A)+\Srel(f|g).
		\end{aligned}
	\end{equation}
	
	\item Subsequently, for type III cases one can construct the crossed product algebras both for the code space and the physical space, which we call the holographic crossed product algebras, such that the algebraic reconstruction theorem admits the refined version of the algebraic RT formula. As a result, the complete algebraic reconstruction theorem which gives a better analog for the holographic setting is as follows,

	\begin{thm}\label{thm:complete alg}
		We start with the following algebraic setup:
		\begin{itemize}
			\item Let $\A_{code},\A_{phys}$ be generic von Neumann factors on $\H_{code},\H_{phys}$ respectively, with $\A'_{code},\A'_{phys}$ respectively being the commutants. Let $V: \H_{code}\rightarrow \H_{phys}$ be an isometry which induces von Neumann factors $V\A_{code}V^\dagger,V\A'_{code}V^\dagger$ on the image of $V$ denoted as $\im(V)=V(\H_{code})\subset \H_{phys}$.

			\item Notations: vectors in $\H_{code}$ and operators in $\A_{code}$ are labeled by the tilde sign, and operators in commutants are labeled by the prime sign. For example, $\ket{\widetilde{\Psi}}\in\H_{code}$, $\widetilde{\O}\in \A_{code},\widetilde{\O}'\in\A'_{code}$. 
			
			\item Consider a cyclic and separating state $\ket{\widetilde{\Omega}}\in\H_{code}$, we construct the crossed product algebras:
			\begin{equation}
				\widehat{\A}_{code;\widetilde{\Omega}}\equiv \A_{code} \rtimes_{\sigma^{\widetilde{\Omega}}} \mathbb{R},\quad \widehat{\A}_{phys;\Omega}\equiv \A_{phys} \rtimes_{\sigma^{\Omega}} \mathbb{R}\quad \text{where}\quad \ket{\Omega}\equiv V\ket{\widetilde{\Omega}}
			\end{equation}
			which act respectively on
			\begin{equation}
				\widehat{\H}_{code}\equiv\H_{code} \otimes L^2(\mathbb{R}),\quad \widehat{\H}_{phys}\equiv\H_{phys} \otimes L^2(\mathbb{R}).
			\end{equation} 
			We call the pair $(\widehat{\A}_{code;\widetilde{\Omega}},\widehat{\A}_{phys;\Omega})$ the holographic crossed product algebras with respect to $\ket{\Omega}$.

			\item Let $\id_{L^{2}(\Rbb)}
			:L^2(\mathbb{R})\rightarrow L^2(\mathbb{R})$ be the identity such that $\widehat{V}=V\otimes \id_{L^{2}(\Rbb)}:\widehat{\H}_{code}\rightarrow \widehat{\H}_{phys}$ is also an isometry. We use $\ket{\widetilde{f}}$ to denote the element in the $L^2(\Rbb)$ part of $\widehat{\H}_{code}$ then $\ket{f}\equiv \overline{V}\ket{\widetilde{f}}$ is the image under $\overline{V}$.
			
			\item Suppose that the set of cyclic and separating vectors w.r.t. $\A_{code}$ is dense in $\H_{code}$ ($\Leftrightarrow$ the set of cyclic and separating vectors w.r.t. $\A'_{code}$ is dense in $\H_{code}$).

			\item Suppose that if $\ket{\widetilde{\Psi}}\in\H_{code}$ is cyclic and separating w.r.t. $\A_{code}$, then $V\ket{\widetilde{\Psi}}\in\H_{phys}$ is cyclic and separating w.r.t. $\A_{phys}$.

			\item We define the semiclassical approximation as assuming both $\widetilde{f}(x)$ and $f(x)$ are slowly varying. 
			
		\end{itemize}
		
		The following statements are equivalent semiclassically: 
		\begin{enumerate}
			\item For any $\widetilde{\O}\in\A_{code},\widetilde{\O}'\in\A_{code}'$, there exist $\O\in \A_{phys},\O'\in \A_{phys}'$ such that
			\begin{equation}
				V\widetilde{\O}\ket{\widetilde{\Psi}}=\O V\ket{\widetilde{\Psi}}, V\widetilde{\O}'\ket{\widetilde{\Psi}}=\O' V\ket{\widetilde{\Psi}},\forall \ket{\widetilde{\Psi}}\in\H_{code}.
			\end{equation}
			
			\item For any $\ket{\widetilde{\Psi}},\ket{\widetilde{\Phi}}\in\H_{code}$ with $\ket{\widetilde{\Psi}},\ket{\widetilde{\Phi}}$ both cyclic and separating w.r.t. $\A_{code}$, 
			\begin{equation}
				\begin{aligned}
					&\Srel(\widetilde{\Psi}|\widetilde{\Phi};\A_{code})=\Srel(\Psi|\Phi;\A_{phys}),\\&\Srel(\widetilde{\Psi}|\widetilde{\Phi};\A'_{code})=\Srel(\Psi|\Phi;\A'_{phys}),
				\end{aligned}
			\end{equation}
			where $\ket{\Psi}\equiv V\ket{\widetilde{\Psi}},\ket{\Phi}\equiv V\ket{\widetilde{\Phi}}$.
			
			\item For any $\P \in\A_{phys},\P '\in\A_{phys}'$,
			\begin{equation}
				V^\dagger \P V\in\A_{code},\quad V^\dagger \P 'V\in\A_{code}'.
			\end{equation}

			\item

			Consider any cyclic and separating $\ket{\widetilde{\Psi}}\in\H_{code}$  w.r.t. $\A_{code}$. If $\A_{code},\A_{phys}$ are von Neumann factors of type I/II, we have
			\begin{equation}
				\begin{aligned}
					&S(\widetilde{\Psi};\A_{code})=S(\Psi;\A_{phys}),\\ &S(\widetilde{\Psi};\A_{code}')=S(\Psi;\A_{phys}').
				\end{aligned}
			\end{equation}
			If $\A_{code},\A_{phys}$ are von Neumann factors of type III, we have\footnote{The notation ``$\approx$'' is manifest by an error at $O(\epsilon)$, as we will see in the main text.}
			\begin{equation}
				\begin{aligned}
					&S(\widetilde{\Psi}_{\widetilde{f}};\widehat{\A}_{code;\widetilde{\Omega}})\approx S(\Psi_{f};\widehat{\A}_{phys;\Omega}),\\ &S(\widetilde{\Psi}_{\widetilde{f}};\widehat{\A'}_{code;\widetilde{\Omega}})\approx S(\Psi_{f};\widehat{\A'}_{phys;\Omega}).
				\end{aligned}
			\end{equation}
		\end{enumerate}

	\end{thm}
\end{itemize}

We end the introduction section by giving some comments about the theorem \ref{thm:complete alg}. The theorem is proved without any physical background. Nevertheless to say, when applying to the holographic setting, the four statements correspond to the algebraic version of the entanglement wedge reconstruction, the JLMS formula, radial commutativity, and the RT formula, respectively. In fact, the first three statements are proved to be equivalent \cite{Kang:2018xqy,Kang:2019dfi} in factors of all types and without assuming the semiclassical approximation, where physically speaking both the observers in the bulk and that in the boundary are approximately unentangled with the quantum field degrees of freedom. The inclusion of the first half of the statement 4 is achieved recently in \cite{Xu:2024xbz} for factors of type I/II only (also without assuming the semiclassical approximation). It is only within the semiclassical description that the other half of the statement 4 regarding the factors of type III, which is more faithful to the holographic framework, can be included into the equivalent relations.

In the main content of the work, we intend to make our discussions explicit and organize the work as follows. In section \ref{sec:pre}, we give a brief introduction to the von Neumann algebra, the modular theory (or Tomita-Takesaki theory) and the algebraic entropies, and also the necessary basics of the crossed product algebra. In section \ref{sec:semicl}, we introduce the semiclassical description, within which we derive several useful relations for operators and prove that the relative entropy of the crossed product algebra factorizes. In section \ref{sec:alg}, we discuss how to generalize the algebraic reconstruction theorem into type III cases. We end with section \ref{sec:discussions} by giving further discussions about the algebraic RT formula in type III cases.

\section{Preliminaries}\label{sec:pre}

In this section, we briefly recall some basic aspects of von Neumann algebra which will be used in what follows. All statements are presented without proof, and one can refer to \cite{Vaughan2009,Sorce:2023fdx,reed1972methods,takesaki2006tomita,Liu:2025krl} for more rigorous arguments, and also \cite{Harlow:2016vwg,Witten:2018zxz,Kang:2018xqy,Kang:2019dfi,Xu:2024xbz,Jensen:2023yxy,Faulkner:2024gst,AliAhmad:2023etg} for relevant contents.

\subsection{Hilbert space, von Neumann algebra and modular theory}

\begin{defn}
	A Hilbert space is a complex vector space $\mathcal{H}$ endowed with the inner product $\langle \cdot|\cdot\rangle:\mathcal{H}\times \mathcal{H}\rightarrow \mathbb{C}$ and the norm $||v||=\sqrt{\braket{v}{v}}$, such that the following properties are satisfied:
	\begin{itemize}
			
			\item Linearity for the second variable: $\langle \xi|\psi_1+\alpha \psi_2\rangle=\langle \xi|\psi_1\rangle+\alpha \langle \xi|\psi_2\rangle$,
			
			\item Hermiticity: $\braket{\psi}{\xi}^*=\braket{\xi}{\psi}$,
			
			\item Positive semi-definiteness: $\braket{\psi}{\psi}\geq 0$,
			
			\item Non-degeneracy: $\braket{\psi}{\psi}=0\Leftrightarrow \ket{\psi}=0$,
			
			\item Completeness: $\H$ is complete (i.e. all Cauchy sequences converge) for the norm.

		\end{itemize}
\end{defn}

An important example of a Hilbert space is the square-integrable space on $\Rbb$, which is defined as
\begin{equation}
	L^2(\Rbb)=\left\{f~\bigg|\int_{\mathbb{R}}|f(x)|dx<\infty\right\}.
\end{equation}
In Dirac notation, we write $\ket{f}\in L^2(\mathbb{R})$ as\footnote{Strictly speaking, $\ket{x}$ and $\ket{p}$ are generalized eigenvectors in the sense of the rigged Hilbert space (Gelfand triple) $\mathcal{S}(\Rbb)\subset L^{2}(\Rbb)\subset\mathcal{S}'(\Rbb)$, where $\mathcal{S}(\Rbb)$ is the Schwartz space of rapidly decaying smooth functions. In the paper, we follow the treatments in appendix E of \cite{Jensen:2023yxy} and all final formulas are evaluated on square-integrable functions where the expressions are well-defined.}
\begin{equation}
	\ket{f}=\int_{\mathbb{R}}dx~f(x)\ket{x}
\end{equation}
where we call $\ket{x}$ the ``position'' basis, then $L^2(\Rbb)$ is endowed with an inner product $\braket{f}{g}=\int_{\mathbb{R}}f^*(x)g(x)dx$ and a norm $||f||=\sqrt{\braket{f}{f}}$. For physics interest, hereafter we only consider normalized element in $L^2(\Rbb)$, i.e. $||f||=1$. One can also work in the ``momentum'' space such that
$\ket{f}=\int_{\mathbb{R}}dp~\F_{f}(p)\ket{p}$ with
\begin{equation}
	\F_{f}(p)=\frac{1}{\sqrt{2\pi}}\int_{\mathbb{R}}dx ~f(x) e^{-ipx},\quad f(x)=\frac{1}{\sqrt{2\pi}}\int_{\mathbb{R}}dp ~\F_{f}(p) e^{ipx},
\end{equation}
and the relations between the two bases are given by
\begin{equation}\label{eq:xp}
	\begin{aligned}
		\braket{x}{x'}=\delta(x-x'),\quad \braket{p}{p'}=\delta(p-p'),\quad \braket{x}{p}=\frac{1}{\sqrt{2\pi}}e^{ipx}.
	\end{aligned}
\end{equation}
We also introduce the corresponding ``position'' operator $X$ and the corresponding ``momentum'' operator $P$ as follows,
\begin{equation}\label{eq:XP rel}
	\begin{aligned}
		&P\ket{p}=p\ket{p},\quad X\ket{x}=x\ket{x},\quad [X,P]=i\\
		&e^{isP}\ket{x}=\ket{x-s},\quad e^{isX}\ket{p}=\ket{p+s}.
	\end{aligned}
\end{equation}


\begin{defn}
	A linear operator on a Hilbert space $\H$ is a linear map from (a subspace of) $\H$ into $\H$. The set of all such operators is denoted by $\L(\H)$.
\end{defn}

\begin{defn}
	A bounded operator is a linear operator $\O$ satisfying $||\O\ket{\psi}||\leq K||\ket{\psi}||,~\forall\ket{\psi}\in\H$ for some $K\in\mathbb{R}$. The infimum of all such $K$ is called the norm of $\O$. The algebra of all bounded operators on $\H$ is denoted by $\B(\H)\subset \L(\H)$.
\end{defn}

\begin{defn}
	The commutant of a subset $S\subset \B(\H)$ is a subset $S'\subset \B(\H)$ defined by
	\begin{equation}
		S'\equiv\{\O\in\B(\H)|[\O,\P]=0,~\forall \P\in S\}
	\end{equation}
	i.e. every element in $S'$ commutes with all elements in $S$.
\end{defn}

\begin{defn}
	The hermitian conjugate (or adjoint) of an operator $\O$ is an operator $\O^\dagger$ satisfying $\braket{\psi}{\O\xi}=\braket{\O^\dagger \psi}{\xi}$. A hermitian (or self-adjoint) operator $\O$ satisfies $\O=\O^\dagger$.
\end{defn}

\begin{defn}
	A von Neumann algebra on $\H$ is a subalgebra $\A\subset\B(\H)$ satisfying
	\begin{itemize}
		\item $\id\in \A$,
		
		\item $\A$ is closed under hermitian conjugation,
		
		\item $\A''=\A$.
		
	\end{itemize} 
\end{defn}


\begin{defn}
	A von Neumann algebra $\A$ is a factor if it has a trivial center $\mathcal{Z}$:
	\begin{equation}
		\mathcal{Z}\equiv \A\cap\A'=\{\lambda \id|\lambda\in\C\}
	\end{equation}
	otherwise $\A$ is called a non-factor.
\end{defn}

Every von Neumann algebra that is not a factor admits a canonical decomposition into factors \cite{Sorce:2023fdx,Harlow:2016vwg}. Thus, the general classification problem reduces to the study of factors, which are exhaustively classified into type I, type II, and type III. The applicability of certain fundamental concepts varies with the type, as summarized in what follows.

\[
\begin{array}{|c|c|c|c|c|c|}
	\hline {\text {Type}} & \H=\H_A \otimes \H_B & \Tr & \rho_\psi & S(\psi;\A)  & S_{\text {rel }}(\psi|\xi;\A) \\
	\hline \text { I} & \checkmark & \checkmark & \checkmark & \checkmark  & \checkmark \\
	\hline \text { II} & \times & \checkmark & \checkmark & \checkmark &  \checkmark \\
	\hline \text { III} & \times & \times & \times & \times  & \checkmark \\
	\hline
\end{array}
\]
where we write $\H=\H_A \otimes \H_B$ to denote the decomposition of the Hilbert space corresponding to distinguished spatial subregions. The quantum relative entropy admits a fully algebraic generalization, denoted by $S_{\text {rel }}(\psi|\xi;\A)$, which is well defined for factors of arbitrary type. In contrast, the algebraic von Neumann entropy $S(\psi;\A)$ is rigorously defined only for type I and type II factors, as type III lacks a well-defined trace class structure and the associated notion of density operators.

A finer algebraic classifications gives factors of type I$_{n}$, type I$_{\infty}$, type II$_{1}$, type II$_{\infty}$, type III$_{0}$, type III$_{\lambda}$ ($0<\lambda<1$), and type III$_{1}$ \cite{Sorce:2023fdx,Liu:2025krl}. Factors of type I$_{n}$ and type II$_{1}$ are finite in the sense that all respective projections are finite, while factors of type I$_{\infty}$ and type II$_{\infty}$ are infinite since at least one of their respective projections is infinite. On the other hand, all type III factors are infinite, and their finer classification depends on the range of the intersection over all spectra of all possible modular operators (introduced later) associated to the factor. In particular, the factor is of type III$_1$ if the set ranges from $[0,\infty)$, which is commonly believed to be the type that a subregion in quantum field theory is associated to, see \cite{Yngvason:2004uh,Fredenhagen:1984dc,Araki:1963klf,Araki:1964lyc} and also \cite{Sorce:2023fdx,Liu:2025krl} for more details.

For type I$_n$ and II$_1$ factors, the trace is finite and a normalized normal tracial state exists, while For type I$_\infty$ and II$_\infty$ factors the trace is faithful normal semifinite, i.e. it is finite on a dense set of ``trace class'' operators  but infinite on the identity, so no normalized tracial state exists. We will return to this issue when defining algebraic entropies in section~\ref{sec:alg entropies}.

\begin{defn}
	A subset $\mathcal{H}_{0}\subset\mathcal{H}$ is dense in $\mathcal{H}$ if for every vector $\ket{\psi} \in \mathcal{H}$ and for every $\epsilon>0$, there exists a vector $\ket{\phi} \in \mathcal{H}_{0}$ such that $\|\ket{\psi}-\ket{\phi}\|<\epsilon$.
\end{defn}

\begin{defn}
	$\ket{\psi}\in\H$ is cyclic with respect to a von Neumann algebra $\A$ if $\A\ket{\psi}\equiv\{\O\ket{\psi}|\forall \O\in \A\}$ is dense in $\H$.
\end{defn}

\begin{defn}
	$\ket{\psi}\in\H$ is separating with respect to a von Neumann algebra $\A$ if $\O\ket{\psi}=0 $ implies $\O=0$ for $\O\in\A$.
\end{defn}

\begin{thm}\label{thm:cs}
	$\ket{\psi}\in\H$ is separating with respect to $\A$ if and only if $\ket{\psi}\in\H$ is cyclic with respect to $\A'$, and vice versa.
\end{thm}


A direct consequence is that when we assume $\ket{\psi}\in\H$ is both cyclic and separating with respect to $\A$, then $\ket{\psi}\in\H$ is also both cyclic and separating with respect to $\A'$.

\begin{defn}
	A relative Tomita operator on $\A$ is an anti-linear operator satisfying
	\begin{equation}\label{def:relative Tomita}
		S_{\xi|\psi}\left(\O\ket{\psi}\right)=\O^\dagger \ket{\xi},\quad \forall \O\in\A
	\end{equation}
\end{defn}

Notice that $S_{\xi|\psi}$ is densely defined (i.e. whose domain is a dense subset of $\H$) if and only if $\ket{\psi}$ is cyclic and separating with respect to $\A$. The cyclic condition ensures that the domain $\O\ket{\psi}$ is dense while the separating condition is to avoid the possibility that $\O\ket{\psi}=0,~ \O^\dagger\ket{\xi}\neq 0$. Hereafter we mostly assume the cyclic separating condition of $\ket{\psi}$ for $S_{\xi|\psi}$.

\begin{thm}
	Provided both $\ket{\psi},\ket{\xi}$ are cyclic and separating with respect to $\A$, we have
	\begin{equation}
		S_{\xi|\psi}^{-1}=S_{\psi|\xi},
	\end{equation}
\end{thm}


\begin{thm}
	Provided $\ket{\psi}$ is cyclic and separating with respect to $\A$, we have
	\begin{equation}
		S^\dagger_{\xi|\psi}=S'_{\xi|\psi}
	\end{equation}
	where $S'_{\xi|\psi}$ is a relative Tomita operator on $\A'$.
\end{thm}

%

\begin{defn}
	Provided $\ket{\psi}$ is cyclic and separating with respect to $\A$, the relative modular operator on $\A$ is defined by\footnote{There is an equivalent definition of the relative modular operator via the unique polar decomposition $S_{\xi|\psi}=J_{\xi|\psi}\Delta_{\xi|\psi}^{\frac{1}{2}}$ with $J_{\xi|\psi}$ anti-unitary.}
	\begin{equation}\label{def:rel mod op}
		\Delta_{\xi|\psi}\equiv S^\dagger_{\xi|\psi}S_{\xi|\psi}
	\end{equation}
	and the relative modular Hamiltonian on $\A$ is defined by 
	\begin{equation}\label{def:rel mod H}
		h_{\xi|\psi}\equiv-\log \Delta_{\xi|\psi}
	\end{equation}
\end{defn}

\begin{thm}
	The relative modular operator $\Delta_{\xi|\psi}$ on $\A$ is Hermitian.
\end{thm}


\begin{thm}
	Provided both $\ket{\psi},\ket{\xi}$ are cyclic and separating with respect to $\A$, we have
	\begin{equation}
		\Delta_{\psi|\xi}^{-1}=\Delta'_{\xi|\psi}\quad\Rightarrow\quad h_{\psi|\xi}=-h_{\xi|\psi}'
	\end{equation}
	where $\Delta'_{\xi|\psi}$ and $h_{\xi|\psi}'$ are the relative modular operator and the relative modular Hamiltonian on $\A'$ respectively.
\end{thm}


\begin{defn}
	Provided $\ket{\psi}$ is cyclic and separating with respect to $\A$, we define the Tomita operator $S_{\psi}\equiv S_{\psi|\psi}:\O\ket{\psi}\mapsto\O^\dagger \ket{\psi}$; the modular operator $\Delta_{\psi}\equiv\Delta_{\psi|\psi}= S^\dagger_{\psi}S_{\psi}$; the modular Hamiltonian $h_{\psi}\equiv h_{\psi|\psi}=-\log \Delta_{\psi}$.
\end{defn}



For the physical applications considered in previous works, e.g. \cite{Chandrasekaran:2022cip,Chandrasekaran:2022eqq,Jensen:2023yxy,Faulkner:2024gst}, the reference state $\ket{\Omega}$, commonly used to construct the crossed product algebra, is a KMS state at inverse temperature $\beta>0$ with the modular operator satisfies
\begin{equation}\label{eq:KMS}
	\Delta_{\Omega}=e^{-\beta h},
\end{equation}
where $h$ is a self-adjoint operator with the physical interpretation of the renormalized energy,. Now we have $h_{\Omega}=\beta h$ and the parameter $\beta$ is set by the specific physical context (e.g. $\beta=\beta_{\text{dS}}$ for the Bunch-Davies state $\ket{\Psi_{\text{dS}}}$ in de~Sitter~\cite{Chandrasekaran:2022cip}). For notational simplicity, we henceforth set $\beta=1$ throughout this paper, which is equivalent to absorbing the inverse temperature into the definition of the modular Hamiltonian so that no separate $\beta$ factor appears. One can restore the original expressions by replacing $h_{\Omega}\to\beta h$.

%


\begin{defn}\label{def:cocycle}
	Let $\A$ be a von Neumann algebra acting on a Hilbert space $\H$. The Connes cocycle associated with two states $\ket{\psi}$ and $\ket{\omega}$ is defined by
	\begin{equation}\label{eq:cocycle-def}
		u_{\psi|\omega}(t)
		\equiv \Delta_{\psi|\omega}^{it}\,\Delta_{\omega}^{-it}
		=\Delta_{\psi}^{it}\,\Delta_{\omega|\psi}^{-it},\quad t\in\Rbb.
	\end{equation}
\end{defn}

The operator $u_{\psi|\omega}(t)$ has the following properties:

\begin{itemize}
	\item $u_{\psi|\omega}(t)$ is a unitary operator on $\H$ for all $t\in\Rbb$.
	
	\item $u_{\psi|\psi}(t)=1$ for all $t\in\Rbb$. 
	
	\item $u_{\psi|\omega}(t)\in\A$ for all $t\in\Rbb$.

	\item For any three cyclic separating states $\ket{\psi},\ket{\phi},\ket{\omega}$, the Connes cocycle satisfies the chain rule:
	\begin{equation} \label{eq:Connes cocycle chain rule}
		u_{\psi|\phi}(t)\,u_{\phi|\omega}(t)=u_{\psi|\omega}(t),\quad \forall t\in\Rbb.
	\end{equation}
	
\end{itemize}

The infinitesimal generator of the Connes cocycle yields the key identity used in the main text. Differentiating \eqref{eq:cocycle-def} at $t=0$,
\begin{equation}\label{eq:cocycle-deriv-app1}
	\begin{aligned}
		\left.\frac{d}{dt}\right|_{t=0}u_{\psi|\omega}(t)
		&= \left.\frac{d}{dt}\right|_{t=0}
		\bigl(e^{-ith_{\psi|\omega}}\,e^{ith_{\omega}}\bigr) \\
		&= -i\,(h_{\psi|\omega}-h_{\omega}),
	\end{aligned}
\end{equation}
and similarly from the second expression in \eqref{eq:cocycle-def},
\begin{equation}\label{eq:cocycle-deriv-app2}
	\left.\frac{d}{dt}\right|_{t=0}u_{\psi|\omega}(t)
	= -i\,(h_{\psi}-h_{\omega|\psi}).
\end{equation}
We therefore obtain
\begin{equation}\label{eq:h-diff}
	h_{\psi|\omega}-h_{\omega}=h_{\psi}-h_{\omega|\psi},
\end{equation}
and both sides are affiliated to $\A$.\footnote{An unbounded operator is affiliated to the algebra if every bounded function of that operator is in the algebra. } Another useful identity is to differentiate the both sides of the chain rule \eqref{eq:Connes cocycle chain rule}:
\begin{equation}
	\begin{aligned}
		\left.\frac{d}{dt}\right|_{t=0}u_{\psi|\phi}(t)
		&= -i\,(h_{\psi}-h_{\phi|\psi}),\\
		\left.\frac{d}{dt}\right|_{t=0}u_{\phi|\omega}(t)
		&= -i\,(h_{\phi|\omega}-h_{\omega})\\
		\left.\frac{d}{dt}\right|_{t=0}u_{\psi|\omega}(t)
		&= -i\,(h_{\psi|\omega}-h_{\omega}),
	\end{aligned}
\end{equation}
where we use \eqref{eq:cocycle-deriv-app2} for $u_{\psi|\phi}(t)$ and \eqref{eq:cocycle-deriv-app1} for $u_{\phi|\omega}(t)$ and $u_{\psi|\omega}(t)$. We therefore have:
\begin{equation}
	\begin{aligned}
		\text{LHS}&=\left[\left.\frac{d}{dt}\right|_{t=0}u_{\psi|\phi}(t)\right]~u_{\phi|\omega}(0)+u_{\psi|\phi}(0)~\left[\left.\frac{d}{dt}\right|_{t=0}u_{\phi|\omega}(t)\right]\\
		&=-i\left(h_{\psi}-h_{\phi|\psi}+h_{\phi|\omega}-h_{\omega}\right),\\
		\text{RHS}&= -i\,(h_{\psi|\omega}-h_{\omega})
	\end{aligned}
\end{equation}
which implies
\begin{equation}\label{eq:chain rule id}
	h_{\phi|\psi}=h_{\phi|\omega}-h_{\psi|\omega}+h_{\psi}.
\end{equation}

\subsection{Algebraic entropies}\label{sec:alg entropies}

In this subsection, we brief introduce the algebraic version of relative entropy and von Neumann entropy. For more related discussions, readers can consult \cite{Chandrasekaran:2022cip,Chandrasekaran:2022eqq,Jensen:2023yxy,Faulkner:2024gst,Xu:2024xbz,Zhong:2026syt}.

The algebraic generalization of the relative entropy is defined as the expectation value of the relative modular Hamiltonian due to Araki \cite{araki1975relative,araki1975inequalities}:
\begin{equation}\label{def:Arakis relative entropy}
	\Srel(\psi|\phi;\A)\equiv\diracprod{\psi}{h_{\phi|\psi}}{\psi}
\end{equation}
with $\ket{\psi}$ being cyclic and separating with respect to the corresponding algebra $\A$. Note that some previous literatures, e.g.\cite{Witten:2018zxz}, use the convention that interchanges the positions of $\psi$ and $\xi$ in the subscript of the relative Tomita operator \eqref{def:relative Tomita}, which results in reversed arguments for the relative entropy. Our notation in \eqref{def:Arakis relative entropy} strictly follows \cite{Witten:2021unn,Chandrasekaran:2022cip,Chandrasekaran:2022eqq,Jensen:2023yxy}. 

In type I/II factors, \eqref{def:Arakis relative entropy} reduces to the familiar expression
	\begin{equation} \label{def:relative entropy}
		\Srel(\psi|\phi;\A)=\Tr\bigl(\rho_{\psi}\log\rho_{\psi}-\rho_{\psi}\log\rho_{\phi}\bigr)=\diracprod{\psi}{\log\rho_{\psi}}{\psi}-\diracprod{\psi}{\log\rho_{\phi}}{\psi},
	\end{equation}
	which is non-negative and vanishes if and only if $\rho_{\psi}=\rho_{\phi}$. The algebraic definition~\eqref{def:Arakis relative entropy} inherits both properties and is well-defined for factors of arbitrary type. See \cite{Witten:2018zxz} for details about how the algebraic relative entropy coincides with the quantum relative entropy in finite-dimensional quantum mechanics, which corresponds to the type I$_n$ cases with $n$ being the dimension of a given system.

We now turn to the algebraic von Neumann entropy, whose definition depends on the type of the factor:
\begin{enumerate}
	\item For a factor of type I$_n$ or II$_1$, the trace is finite on the whole algebra, and $\Tr_{\A}(\id)$ can be normalized to $1$. In this case the trace arises from a normalized normal tracial state $\ket{\tau}$ satisfying
	\begin{equation}\label{eq:tracial-state}
		\diracprod{\tau}{\O\P}{\tau}=\diracprod{\tau}{\P\O}{\tau},\quad \forall\O,\P\in\A\quad\Rightarrow\quad
		\Tr_{\A}(\O)\equiv\diracprod{\tau}{\O}{\tau}.
	\end{equation}
	such that the algebraic version of the von Neumann entropy is defined by
	\begin{equation}\label{def:algvnentropy-finite}
	S(\psi;\A)\equiv-\Srel(\psi|\tau;\A)=-\diracprod{\psi}{h_{\tau|\psi}}{\psi}.
	\end{equation}
	which is the formula used in~\cite{segal1960note,ohya2004quantum,Longo:2022lod}.
	
	\item For a factor of type I$_\infty$ or II$_\infty$, the trace is only semifinite such that $\Tr_{\A}(\id)=\infty$, so no normalized tracial vector exists. Nevertheless, the semifinite trace is sufficient to define density operators as positive elements $\rho_{\psi}$ affiliated to $\A$ satisfying $\Tr_{\A}(\rho)=1$, and the von Neumann entropy is then defined via \cite{Jensen:2023yxy,Faulkner:2024gst}
	\begin{equation}
		S(\rho_{\psi};\A)\equiv-\Tr_{\A}(\rho_{\psi}\log\rho_{\psi})=- \diracprod{\psi}{\log\rho_{\psi}}{\psi}.
	\end{equation}
\end{enumerate}

In both cases, the algebraic von Neumann entropy is uniquely defined up to an overall additive constant. Such constant does not cause ambiguity, since only the difference between entropies has physical meaning. The non-existence of algebraic von Neumann entropy in factors of type III is due to the non-existence of a faithful normal semifinite trace.

\subsection{Crossed product algebra}\label{sec:crossed prod}

%
%

Consider a von Neumann algebra $\A$ with a corresponding Hilbert space $\H$ and a cyclic and separating state $\ket{\Omega}$ with a modular flow $\sigma^{\Omega}$ defined by:
\begin{equation}
	\sigma^{\Omega}_s(a)=e^{ish_{\Omega}}ae^{-ish_{\Omega}},\quad a\in\A,~\forall s\in\mathbb{R}
\end{equation}
where $h_{\Omega}$ is the modular Hamiltonian of $\ket{\Omega}$ with respect to $\A$, we construct the product Hilbert space
\begin{equation}
	\widehat{\H}\equiv\H \otimes L^2(\mathbb{R})=\left\{\ket{\Psi_f}\equiv\ket{\Psi}\otimes\ket{f}\right\},
\end{equation} 
with the crossed product algebra
\begin{equation}\label{eq:crossed product alg}
	\widehat{\A}_{\Omega}\equiv \A \rtimes_{\sigma^{\Omega}} \mathbb{R}=\left\{a_{P},\mathbbm{1}\otimes e^{itX}\big|a\in \A,t\in\mathbb{R}\right\}'',\quad a_{P}\equiv e^{ih_{\Omega}\otimes P}(a\otimes\mathbbm{1})e^{-ih_{\Omega}\otimes P},
\end{equation}
where the double commutant is to ensure that $\widehat{\A}_{\Omega}$ is a von Neumann algebra. There exists an equivalent definition of crossed product algebra which works as the subalgebra of $\A \otimes \B(L^2(\mathbb{R}))$ fixed under the flow generated by $h_{\Omega}\otimes \id+\id\otimes X$ (we write $h_{\Omega}+ X$ for simplicity) \cite{Jensen:2023yxy}:
\begin{equation}
	\widehat{\A}_{\Omega}\equiv \{\widehat{a}\in\A \otimes \B(L^2(\mathbb{R}))\big|e^{is(h_{\Omega}+ X)}\widehat{a}e^{-is(h_{\Omega}+ X)}=\widehat{a},\forall s\in\Rbb\}.
\end{equation}
In this definition, one can regard $L^2(\mathbb{R})$ as an auxiliary register which keeps track of some ``modular time'' to the original algebra $\A$, and the crossed product construction is equivalent to imposing the gauge symmetry $h_{\Omega}+ X$ to the product algebra $\A \otimes \B(L^2(\mathbb{R}))$. Previous results \cite{Chandrasekaran:2022cip,Chandrasekaran:2022eqq,Jensen:2023yxy} regarding quantum gravity introduce an auxiliary observer to the corresponding region whose degrees of freedom are characterized by $L^2(\mathbb{R})$ with $X$ being the observer Hamiltonian (i.e. $s\in \mathbb{R}$ is the time measured by the observer), while $h_{\Omega}$ corresponds to the geometrical constraints such that $h_{\Omega}+ X$ generates an invariant flow in the original algebra coupled with the observers' degrees of freedom. 

There is a useful theorem stating that for two cyclic separating states $\ket{\Omega_1}$ and $\ket{\Omega_2}$, the crossed product algebras are isometric: $\widehat{\A}_{\Omega_1}\simeq \widehat{\A}_{\Omega_2}$ \cite{Witten:2021unn}. In this sense, the crossed product algebra is unique up to an isometry, while practically one would prefer to use some specific cyclic separating state for either simplifying cumbersome calculations or concerning some physical situations.

The key advantage of applying the crossed product construction to type III factors is that, the crossed product algebra is of type II$_\infty$ \cite{takesaki1973duality,connes1973classification,connes2008noncommutative} which possesses a faithful normal semifinite trace, such that one is able to define a trace and also the density operators for the crossed product algebra.\footnote{For the de Sitter case with a positive-energy observer, the crossed product algebra is of type II$_1$~\cite{Chandrasekaran:2022cip}; while for black hole spacetimes it is in general type II$_\infty$~\cite{Witten:2021unn,Chandrasekaran:2022eqq}. Both cases are covered by the semifinite trace framework.} Following \cite{Jensen:2023yxy} (with $\beta=1$), we define the trace for the crossed product algebra as follows,
\begin{equation}
	\widehat{\Tr}(\widehat{a})\equiv 2\pi \left(\bra{\Omega}\otimes\bra{0}_{P}\right)e^{X}\widehat{a}\left(\ket{\Omega}\otimes\ket{0}_{P}\right),\quad \forall \widehat{a}\in\widehat{\A}_{\Omega}
\end{equation}
where $\ket{0}_{P}$ is the zero-momentum eigenstate of $L^2(\mathbb{R})$. In the position basis, one can apply \eqref{eq:xp} and find that
\begin{equation}
	\begin{aligned}
		\widehat{\Tr}(\widehat{a})=& 2\pi \left(\bra{\Omega}\otimes\bra{0}_{P}\right)e^{X}\widehat{a}\left(\ket{\Omega}\otimes\ket{0}_{P}\right)\\
		=&2\pi \int dxdx' \left(\bra{\Omega}\otimes\bra{0}_{P}\right)\left(\ketbra{x}{x}\right)e^{X}\widehat{a}\left(\ketbra{x'}{x'}\right)\left(\ket{\Omega}\otimes\ket{0}_{P}\right)\\
		=&\int dxdx'\bra{\Omega}\bra{x}e^{X}\widehat{a}\ket{x'}\ket{\Omega}\\
		=&\int dxdx'\delta(x-x')e^{x}\bra{\Omega}\widehat{a}\ket{\Omega}\\
		=&\int dx ~e^{x}\bra{\Omega}\widehat{a}\ket{\Omega}
	\end{aligned}
\end{equation}
which coincides with an alternative definition of the trace function in \cite{Witten:2021unn,AliAhmad:2023etg}. For a state $\ket{\Psi_{f}}\in\widehat{\H}$, the density operator $\rho_{\Psi_{f}}$ affiliated to $\widehat{\A}_{\Omega}$ is then defined by the relation
\begin{equation}\label{eq:rho-trace}
	\widehat{\Tr}\bigl(\rho_{\Psi_{f}}\,\widehat{a}\bigr)
	=\diracprod{\Psi_{f}}{\widehat{a}}{\Psi_{f}},\quad\forall\widehat{a}\in\widehat{\A}_{\Omega}.
\end{equation}


%

\section{Semiclassical description}\label{sec:semicl}

In this section, we demonstrate how to practically apply the semiclassical approximation in the presence of crossed product algebra, and one can consult \cite{Chandrasekaran:2022cip,Chandrasekaran:2022eqq,Jensen:2023yxy} for more details. After imposing the approximation, we then prove that the relative entropy between two states in the crossed product algebra approximately factorizes into a sum of two relative entropies: the relative entropy intrinsic to the original algebra, and the relative entropy quantifying the distinguishability between the respective observers.

\subsection{Semiclassical states}\label{sec:semicl-states}

While promoting the algebra into the crossed product algebra via introducing an auxiliary observer with a wavefunction $f\in L^2(\mathbb{R})$, $\ket{\Psi}$ is promoted into $\ket{\Psi_f}\equiv\ket{\Psi}\otimes\ket{f}$ and the condition of semiclassical approximation is to select a special class of observers whose degrees of freedom are not strongly entangled with the quantum field degrees of freedom. One can think of it as follows, the degrees of freedom between the observer and the state, which are presented as operators in the respective algebras in the algebraic language, are entangled due to the conjugation $e^{ih_{\Omega}\otimes P}(a\otimes\mathbbm{1})e^{-ih_{\Omega}\otimes P}$ in the crossed product algebra \eqref{eq:crossed product alg}. Therefore, we can require that $f(x)$ is slowly varying such that due to the uncertainty principle its Fourier mode $\F_{f}(p)$ is sharply peaked around some specific momentum, which we can set to be $p=0$ without loss of generality.

Our practical treatments follow \cite{Chandrasekaran:2022cip,Chandrasekaran:2022eqq}. We first consider a state $\ket{\Psi_{f}}\subset\widehat{\H}$ and introduce a dimensionless infinitesimal parameter $\epsilon$ with $0<\epsilon\ll 1$, which is to parameterize the degree to which the observer's degrees of freedom are entangled with the quantum field degrees of freedom. Note the standard treatments require $\epsilon<\beta$ while we have set $\beta=1$ in the paper. We then define the semiclassical states to take the form of $\ket{\Psi_{f}}\subset\widehat{\H}$ with
\begin{equation}\label{eq:f-param}
	f(x)=\epsilon^{1/2}\,k(\epsilon x),
\end{equation}
where $k\in L^2(\Rbb)$ is a also smooth, bounded normalized function such that
\begin{equation}
	\int_{-\infty}^{\infty} dx\,\abs{f(x)}^{2}=\int_{-\infty}^{\infty} dy\,\abs{k(y)}^{2}=1.
\end{equation}
The physical interpretation is that $\abs{f(x)}^{2}$ is mostly supported for $x\sim O(1/\epsilon)$, so that in momentum space the Fourier transform $\F_{f}(p)$ is sharply peaked around $p=0$ with width $O(\epsilon)$ due to the uncertainty principle, as required by the semiclassical approximation. 

The key observation is that for semiclassical states, the conjugation by $e^{ih_{\Omega}\otimes P}$ in \eqref{eq:crossed product alg} approximately acts as the identity when evaluated in expectation values. To be a bit qualitative,
\begin{equation}\label{eq:momentum-bound}
	\|P\ket{f}\|
	=\Bigl(\int_{-\infty}^{\infty} dp\;p^{2}\abs{\F_{f}(p)}^{2}\Bigr)^{1/2}
	=\epsilon\;\Bigl(\int_{-\infty}^{\infty} dq\;q^{2}\abs{\F_{k}(q)}^{2}\Bigr)^{1/2}
	=O(\epsilon),
\end{equation}
which implies that $\big\| \bigl(\id-e^{-ih_{\Omega}\otimes P}\bigr)\ket{\Psi_{f}} \big\|
= O(\epsilon)$. We emphasize that the above analysis is a state-level statement about matrix elements between semiclassical states. It involves neither an algebraic-level approximation of the von Neumann algebra $\widehat{\A}_{\Omega}$ by a tensor product algebra, nor a modular-level approximation in which the modular operator factorizes.

\subsection{Factorization of relative entropy}\label{sec:factorization}

As already shown in section 3 of \cite{Chandrasekaran:2022cip} (with $\beta=1$), the approximate density operator for the semiclassical state $\ket{\Psi_f}$ is given by
	\begin{equation}\label{eq:rho-approx}
		\rho_{\Psi_{f}}
		=\overline{f}\bigl(X+h_{\Omega}\bigr)\,
		e^{-X}\,\Delta_{\Psi|\Omega}\,
		f\bigl(X+h_{\Omega}\bigr)+O(\epsilon)=\abs{f\bigl(X+h_{\Omega}\bigr)}^{2}\,
		e^{-X}\,\Delta_{\Psi|\Omega}+O(\epsilon).
	\end{equation}
Taking the logarithm of \eqref{eq:rho-approx}, we then obtain
\begin{equation}\label{eq:-log-rho-Psi}
	-\log\rho_{\Psi_{f}} =h_{\Psi|\Omega}+X-\log\abs{f(X)}^{2}+O(\epsilon),
\end{equation}
where we use the fact $f$ is slowly varying $f(X+h_{\Omega})\approx f(X)$. Now consider another semiclassical state $\ket{\Phi_{g}}$ with
\begin{equation}\label{eq:-log-rho-Phi}
	-\log\rho_{\Phi_{g}} =
	h_{\Phi|\Omega}+X-\log\abs{g(X)}^{2}+O(\epsilon),
\end{equation}
and recall the relative entropy in type II factors can take the form of \eqref{def:relative entropy}, we have
\begin{equation}\label{eq:cal1}
	\begin{aligned}
	\Srel(\Psi_{f}|\Phi_{g};\widehat{\A}_{\Omega})
&=\diracprod{\Psi_{f}}{\log\rho_{\Psi_{f}}}{\Psi_{f}}-\diracprod{\Psi_{f}}{\log\rho_{\Phi_{g}}}{\Psi_{f}}\\
&=\diracprod{\Psi}{\left(h_{\Phi|\Omega}-h_{\Psi|\Omega}\right)}{\Psi}+\int_{\Rbb} dx\,\abs{f(x)}^{2}\,
\log\frac{\abs{f(x)}^{2}}{\abs{g(x)}^{2}}+O(\epsilon)\\&=\diracprod{\Psi}{\left(h_{\Phi|\Omega}-h_{\Psi|\Omega}\right)}{\Psi}+\Srel(f|g)+O(\epsilon)
	\end{aligned}
\end{equation}
where $\Srel(f|g)=\int_{\Rbb} dx\,\abs{f(x)}^{2}\,
\log\frac{\abs{f(x)}^{2}}{\abs{g(x)}^{2}}$ is the classical relative entropy between the probability densities $\abs{f}^{2}$ and $\abs{g}^{2}$ in $L^{2}(\Rbb)$. To proceed, we make use of the Connes cocycle identity \eqref{eq:chain rule id}:
\begin{equation}
	h_{\Phi|\Psi}=h_{\Phi|\Omega}-h_{\Psi|\Omega}+h_{\Psi},
\end{equation}
then we have
\begin{equation}
	\diracprod{\Psi}{\left(h_{\Phi|\Omega}-h_{\Psi|\Omega}\right)}{\Psi}=\diracprod{\Psi}{h_{\Phi|\Psi}}{\Psi}-\diracprod{\Psi}{h_{\Psi}}{\Psi}=\Srel(\Psi|\Phi;\A)-0.
\end{equation}
Substituting it into \eqref{eq:cal1}, we arrive at
\begin{equation} 
	\begin{aligned}
		\Srel(\Psi_{f}|\Phi_{g};\widehat{\A}_{\Omega})
		=\Srel(\Psi|\Phi;\A)+\Srel(f|g)+O(\epsilon)
	\end{aligned},
\end{equation}
which is the desired result.

For the general case $f\neq g$, the cross term $\Srel(f|g)$ is the classical distinguishability between the two observers' wavefunctions. When the wavefunctions are identical, the relative entropy in the crossed product algebra reduces, at leading order in $\epsilon$, to the relative entropy in the original algebra,
\begin{equation}\label{eq:rel-ent-same-f}
	\Srel(\Psi_{f}|\Phi_{f};\widehat{\A}_{\Omega})
	=\Srel(\Psi|\Phi;\A)+O(\epsilon),
\end{equation}
which matches the expectation that, within the semiclassical regime, the observer merely serves as an algebraic ``regulator''~\cite{Kudler-Flam:2023hkl,Kudler-Flam:2023qfl} that renders the type~III factors to type~II factors while leaving the distinguishability of the underlying quantum field states intact.

\section{Algebraic reconstruction theorem in type III cases}\label{sec:alg}

Now we are ready to extend the algebraic reconstruction theorem to type III cases via the crossed product method. We first review the algebraic reconstruction theorem in type I/II cases, see \cite{Kang:2018xqy,Kang:2019dfi,Xu:2024xbz} for more details. Our setup is as follows:
\begin{itemize}
	\item Let $\A_{code},\A_{phys}$ be von Neumann factors of type I/II on $\H_{code},\H_{phys}$ respectively, with $\A'_{code},\A'_{phys}$ respectively being the commutants. Let $V: \H_{code}\rightarrow \H_{phys}$ be an isometry which induces von Neumann factors $V\A_{code}V^\dagger,V\A'_{code}V^\dagger$ of type I/II on the image of $V$ denoted as $\im(V)=V(\H_{code})\subset \H_{phys}$.

	\item Notations: vectors in $\H_{code}$ and operators in $\A_{code}$ are labeled by the tilde sign, and operators in commutants are labeled by the prime sign. For example, $\ket{\widetilde{\Psi}}\in\H_{code}$, $\widetilde{\O}\in \A_{code},\widetilde{\O}'\in\A'_{code}$. (For simplicity, relative operators on $\H_{code}$ are not labeled by the tilde sign, but one can tell from the states they apply.)

	\item Suppose that the set of cyclic and separating vectors w.r.t. $\A_{code}$ is dense in $\H_{code}$ ($\Leftrightarrow$ the set of cyclic and separating vectors w.r.t. $\A'_{code}$ is dense in $\H_{code}$).

	\item Suppose that if $\ket{\widetilde{\Psi}}\in\H_{code}$ is cyclic and separating w.r.t. $\A_{code}$, then $V\ket{\widetilde{\Psi}}\in\H_{phys}$ is cyclic and separating w.r.t. $\A_{phys}$.

\end{itemize}

\begin{thm}\label{thm:alg}
	The following statements are equivalent \footnote{One can see that the theorem is ``symmetric'' between the algebras and the corresponding commutants, in the sense that the operator relations in the algebras have their counterparts in the commutants. Therefore, the theorem sometimes has a suffix ``with complementary recovery'' in the literatures, see e.g. \cite{Harlow:2016vwg}.}: 
	\begin{enumerate}
		\item For any $\widetilde{\O}\in\A_{code},\widetilde{\O}'\in\A_{code}'$, there exist $\O\in \A_{phys},\O'\in \A_{phys}'$ such that
		\begin{equation}
			V\widetilde{\O}\ket{\widetilde{\Psi}}=\O V\ket{\widetilde{\Psi}}, V\widetilde{\O}'\ket{\widetilde{\Psi}}=\O' V\ket{\widetilde{\Psi}},\forall \ket{\widetilde{\Psi}}\in\H_{code}.
		\end{equation}
		
		\item For any $\ket{\widetilde{\Psi}},\ket{\widetilde{\Phi}}\in\H_{code}$ with $\ket{\widetilde{\Psi}},\ket{\widetilde{\Phi}}$ both cyclic and separating w.r.t. $\A_{code}$, 
		\begin{equation}\label{eq:JLMS}
			\begin{aligned}
				&\Srel(\widetilde{\Psi}|\widetilde{\Phi};\A_{code})=\Srel(\Psi|\Phi;\A_{phys}),\\&\Srel(\widetilde{\Psi}|\widetilde{\Phi};\A'_{code})=\Srel(\Psi|\Phi;\A'_{phys}),
			\end{aligned}
		\end{equation}
		where $\ket{\Psi}\equiv V\ket{\widetilde{\Psi}},\ket{\Phi}\equiv V\ket{\widetilde{\Phi}}$.
		
		\item For any $\P \in\A_{phys},\P '\in\A_{phys}'$,
		\begin{equation}
			V^\dagger \P V\in\A_{code},\quad V^\dagger \P 'V\in\A_{code}'.
		\end{equation}

		\item

		For any cyclic and separating $\ket{\widetilde{\Psi}}\in\H_{code}$  w.r.t. $\A_{code}$, we have
		\begin{equation}\label{eq:statement 4}
			\begin{aligned}
				&S(\widetilde{\Psi};\A_{code})=S(\Psi;\A_{phys}),\\ &S(\widetilde{\Psi};\A_{code}')=S(\Psi;\A_{phys}').
			\end{aligned}
		\end{equation}
	\end{enumerate}
\end{thm}


The equivalence among the first three statements holds for factors of all types without any approximation~\cite{Kang:2018xqy,Kang:2019dfi}. The equivalence with statement 4 was proved in~\cite{Xu:2024xbz} for type~I/II factors only, since type~III factors lack a well-defined von~Neumann entropy (section~\ref{sec:alg entropies}). Our task is to extend statement 4 to type~III by promoting the algebras to their crossed products and invoking the semiclassical approximation developed in section~\ref{sec:semicl}. To be explicit, in type III cases we have 
\begin{itemize}
	\item Let $\A_{code},\A_{phys}$ be von Neumann factors of type III on $\H_{code},\H_{phys}$ respectively, with $\A'_{code},\A'_{phys}$ respectively being the commutants. Let $V: \H_{code}\rightarrow \H_{phys}$ be an isometry.
	
\end{itemize}

Then we introduce additional setups as follows:
\begin{itemize}

	\item Consider a cyclic and separating state $\ket{\widetilde{\Omega}}\in\H_{code}$, we construct the crossed product algebras:
	\begin{equation}
		\widehat{\A}_{code;\widetilde{\Omega}}\equiv \A_{code} \rtimes_{\sigma^{\widetilde{\Omega}}} \mathbb{R},\quad \widehat{\A}_{phys;\Omega}\equiv \A_{phys} \rtimes_{\sigma^{\Omega}} \mathbb{R}\quad \text{where}\quad \ket{\Omega}\equiv V\ket{\widetilde{\Omega}}
	\end{equation}
	which act respectively on
	\begin{equation}
		\widehat{\H}_{code}\equiv\H_{code} \otimes L^2(\mathbb{R}),\quad \widehat{\H}_{phys}\equiv\H_{phys} \otimes L^2(\mathbb{R}).
	\end{equation} 
	We call the pair $(\widehat{\A}_{code;\widetilde{\Omega}},\widehat{\A}_{phys;\Omega})$ the holographic crossed product algebras with respect to $\ket{\Omega}$.

	\item Let $\overline{V}
	:L^2(\mathbb{R})\rightarrow L^2(\mathbb{R})$ be an isometry such that $\widehat{V}=V\otimes \overline{V}:\widehat{\H}_{code}\rightarrow \widehat{\H}_{phys}$ is also an isometry. We use $\ket{\widetilde{f}}$ to denote the element in the $L^2(\Rbb)$ part of $\widehat{\H}_{code}$ then $\ket{f}\equiv \overline{V}\ket{\widetilde{f}}$ is the image under $\overline{V}$.

	\item The semiclassical approximation: a wavefunction is semiclassical if it admits the parameterization $f(x)=\epsilon^{1/2}k(\epsilon x)$ with $0<\epsilon\ll 1$ and $k\in L^2(\Rbb)$. We assume throughout this section that both $\widetilde{f}$ and $f=\overline{V}\widetilde{f}$ are semiclassical wavefunctions with the same small parameter $\epsilon$, and that their probability densities satisfy
	\begin{equation}\label{eq:match-density}
		\begin{aligned}
			&\abs{\widetilde{f}(x)}^{2}=\abs{f(x)}^{2}+O(\epsilon),\quad\forall x\in\Rbb.\\
			&\int_{\Rbb} \abs{\widetilde{f}(x)}^{2} dx=\int_{\Rbb} \abs{f(x)}^{2} dx=1,
		\end{aligned}
	\end{equation}
	such that $\widetilde{f}\in L^2(\Rbb)$. For physical interests, we hereafter choose \eqref{eq:match-density} to take
	\begin{equation}\label{eq:match-density2}
		\overline{V}=\id_{L^{2}(\Rbb)}\quad \Rightarrow\quad \widetilde{f}=f,
	\end{equation}
	which is the case of a single observer whose wavefunction is unaffected by the code-to-physical isometry.
\end{itemize}

\begin{thm}\label{thm:alg2}
	The following statements are equivalent up to $O(\epsilon)$	: 
	\begin{enumerate}
		\item For any $\ket{\widetilde{\Psi}},\ket{\widetilde{\Phi}}\in\H_{code}$ with $\ket{\widetilde{\Psi}},\ket{\widetilde{\Phi}}$ both cyclic and separating w.r.t. $\A_{code}$, 
		\begin{equation}\label{eq:eq1}
			\begin{aligned}
				&\Srel(\widetilde{\Psi}|\widetilde{\Phi};\A_{code})=\Srel(\Psi|\Phi;\A_{phys}),\\&\Srel(\widetilde{\Psi}|\widetilde{\Phi};\A'_{code})=\Srel(\Psi|\Phi;\A'_{phys}),
			\end{aligned}
		\end{equation}
		where $\ket{\Psi}\equiv V\ket{\widetilde{\Psi}},\ket{\Phi}\equiv V\ket{\widetilde{\Phi}}$.
		
		\item For any $\ket{\widetilde{\Psi}},\ket{\widetilde{\Phi}}\in\H_{code}$ with $\ket{\widetilde{\Psi}},\ket{\widetilde{\Phi}}$ both cyclic and separating w.r.t. $\A_{code}$, 
		\begin{equation}\label{eq:eq2}
			\begin{aligned}
				&\Srel(\widetilde{\Psi}_{\widetilde{f}}|\widetilde{\Phi}_{\widetilde{f}};\widehat{\A}_{code;\widetilde{\Omega}})=\Srel(\Psi_{f}|\Phi_{f};\widehat{\A}_{phys;\Omega})+O(\epsilon),\\&\Srel(\widetilde{\Psi}_{\widetilde{f}}|\widetilde{\Phi}_{\widetilde{f}};\widehat{\A'}_{code;\widetilde{\Omega}})=\Srel(\Psi_{f}|\Phi_{f};\widehat{\A'}_{phys;\Omega})+O(\epsilon).
			\end{aligned}
		\end{equation}

		\item

		For any cyclic and separating $\ket{\widetilde{\Psi}}\in\H_{code}$ w.r.t. $\A_{code}$, we have
		\begin{equation}\label{eq:statement 4}
			\begin{aligned}
				&S(\widetilde{\Psi}_{\widetilde{f}};\widehat{\A}_{code;\widetilde{\Omega}})=S(\Psi_{f};\widehat{\A}_{phys;\Omega})+O(\epsilon),\\ &S(\widetilde{\Psi}_{\widetilde{f}};\widehat{\A'}_{code;\widetilde{\Omega}})=S(\Psi_{f};\widehat{\A'}_{phys;\Omega})+O(\epsilon).
			\end{aligned}
		\end{equation}
	\end{enumerate}
\end{thm}

\begin{proof}
	We prove the equivalence of the three statements.
	\begin{itemize}
		\item (1) $\Leftrightarrow$ (2).
		For the code side, the semiclassical factorization \eqref{eq:rel-ent-same-f} applied to the crossed product with the same observer wavefunction in both states gives
		\begin{equation}\label{eq:code-side}
			\Srel(\widetilde{\Psi}_{\widetilde{f}}|\widetilde{\Phi}_{\widetilde{f}};\widehat{\A}_{code;\widetilde{\Omega}})
			=\Srel(\widetilde{\Psi}|\widetilde{\Phi};\A_{code})+O(\epsilon).
		\end{equation}
		Similarly, for the physical side,
		\begin{equation}\label{eq:phys-side}
			\Srel(\Psi_{f}|\Phi_{f};\widehat{\A}_{phys;\Omega})
			=\Srel(\Psi|\Phi;\A_{phys})+O(\epsilon).
		\end{equation}
		An analogous factorization holds for the commutant algebras $\A'_{code}$ and $\A'_{phys}$.
		
		From~\eqref{eq:code-side}--\eqref{eq:phys-side}, it follows that if $\Srel(\widetilde{\Psi}|\widetilde{\Phi};\A_{code})=\Srel(\Psi|\Phi;\A_{phys})$, then the two crossed product relative entropies $\Srel(\widetilde{\Psi}_{\widetilde{f}}|\widetilde{\Phi}_{\widetilde{f}};\widehat{\A}_{code;\widetilde{\Omega}})$ and $\Srel(\Psi_{f}|\Phi_{f};\widehat{\A}_{phys;\Omega})$ differ by at most $O(\epsilon)$. Conversely, if the crossed product relative entropies agree to $O(\epsilon)$, the exact relative entropies in the original algebras must agree because the $O(\epsilon)$ errors on both sides are controlled by the same behavior of \eqref{eq:match-density2} and therefore cancel in $\epsilon\rightarrow0$ limit. 
		
		\item (2) $\Leftrightarrow$ (3). Since both $\widehat{\A}_{code;\widetilde{\Omega}}$ and $\widehat{\A}_{phys;\Omega}$ are type~II  factors, theorem~\ref{thm:alg} applies and gives an exact equivalence (1)$\vert_{\epsilon=0}$ $\Leftrightarrow$ (2)$\vert_{\epsilon=0}$ relations, which corresponds to the case that the observer decouples from the quantum field in the semiclassical limit. Each side deviates from its $\epsilon=0$ limit by $O(\epsilon)$ errors that are controlled by the same parameter according to \eqref{eq:match-density2}, so the $O(\epsilon)$-approximate versions are also equivalent.
	\end{itemize}

\end{proof}

To this end, combining the theorems \ref{thm:alg} and \ref{thm:alg2}, we obtain the complete algebraic reconstruction theorem \ref{thm:complete alg} as stated in the introduction section.

\section{Discussions}\label{sec:discussions}

In this work, we first derived several useful relations between operators in the crossed product algebra and that in the original algebra within the semiclassical description, and proved that the relative entropy in the crossed product algebra factorizes into the relative entropies in the original algebra and between the observers, respectively. Second, by applying the crossed product algebra to the algebraic reconstruction theorem \ref{thm:alg} for type I/II factors, we promote the factors to be generic and define the corresponding holographic crossed product algebras such that semiclassically the algebraic RT formula in type III cases are included in the theorem \ref{thm:alg2}, giving a complete algebraic reconstruction theorem \ref{thm:complete alg}.

Readers may be confused about in what sense do the algebraic RT formulae represent the algebraic generalization of the familiar RT formula, as naively there exists no area terms appearing. On the one hand, we would like to emphasize that the naive RT formula is divergent on both sides due to the fact that we are dealing with quantum field theories with type III factors, such that the usual von Neumann entropies are in fact mathematically illegal. Nevertheless, the UV behavior of the area term on RHS coincides with the leading contribution of that of the boundary von Neumann entropy on LHS. On the other hand, the algebraic von Neumann entropy in semiclassical quantum gravity actually corresponds
to the generalized entropy instead of the naive von Neumann entropy of the matter field \cite{Kudler-Flam:2023qfl,Jensen:2023yxy,Faulkner:2024gst}. E.g. authors in \cite{Jensen:2023yxy} have argued that under certain geometric constraints, the algebraic entropy for a semiclassical bulk state $\Phi$ over some part $\Sigma$ of a Cauchy slice $\Sigma_c=\Sigma \cup \bar{\Sigma}$ takes the form of
\begin{equation}\label{eq:Sgen}
	S\left(\rho_{{\Phi}_f}\right)=\left\langle\frac{A_\gamma}{4 G_N}\right\rangle_{{\Phi}_f}+S_{\Phi}^{\mathrm{QFT}}+S_f^{\mathrm{obs}}+\text {const.}=S_{\text {gen}}+\text {const.},
\end{equation}
where $S_{\Phi}^{\mathrm{QFT}}=-\operatorname{Tr}\left(\rho_{\Phi} \log \rho_{\Phi}\right)$ is the usual von Neumann entropy of $\Phi$, and $S_f^{\mathrm{obs}}=-\int_{-\infty}^{\infty} d q|f(q)|^2 \log |f(q)|^2$ is the entropy associated with the probability distribution derived from the observer’s wavefunction, and $A_\gamma$ corresponds to the area of the entangling surface $\gamma=\Sigma\cap\hat{\Sigma}$. In this sense, one can regard the crossed product construction for the entanglement entropy as an algebraic ``regulator'' \footnote{See also \cite{Kudler-Flam:2023hkl} where authors have formally introduced the crossed product construction as a covariant regulator in a quantum field theory under a general curved spacetime.}. Naively, if one can apply \eqref{eq:Sgen} to both the bulk and the boundary in the holographic setting, we expect that in type III cases the algebraic RT formulae take the form of
\begin{equation}\label{eq:alg RT2}
S_{\Psi}^{\mathrm{QFT}}=\left\{\left\langle\frac{A_\gamma}{4 G_N}\right\rangle_{{\widetilde{\Psi}}_f}-\left\langle\frac{A_{\partial\gamma}}{4 G_N}\right\rangle_{{\Psi}_f}\right\}+S_{\widetilde{\Psi}}^{\mathrm{QFT}}+\text{const.},
\end{equation}
where $\gamma$ is the entangling surface between the region associated to $\widetilde{\Psi}$ and its complementary in the bulk while $\partial \gamma$ is the entangling surface between the region associated to $\Psi$ and its complementary on the boundary, and one can see that the area term appears. Despite the similar form, it is not yet known how to generalize \eqref{eq:Sgen} to the holographic framework \footnote{See also section 6.2 in \cite{Jensen:2023yxy}.}, and the extremization procedure as in the Lewkowycz-Maldacena (LM) method \cite{Lewkowycz:2013nqa} to derive \eqref{eq:alg RT2} is also worth further exploration. These topics are certainly rich in physical significance, but they are far beyond the scope of the work, and we leave it for future investigation.

\appendix

\section*{}

\acknowledgments

HZ thanks Qiang Wen, Mingshuai Xu, Yiwei Zhong for helpful discussions. HZ is supported by SEU Innovation Capability Enhancement Plan for Doctoral Students (Grant No. CXJH\_SEU 24137) and the Shing-Tung Yau Center of Southeast University.

\bibliographystyle{JHEP}
\bibliography{bib}

\providecommand{\href}[2]{#2}\begingroup\raggedright\begin{thebibliography}{10}

\bibitem{Witten:2021unn}
E.~Witten, \emph{{Gravity and the crossed product}},
  \href{https://doi.org/10.1007/JHEP10(2022)008}{\emph{JHEP} {\bfseries 10}
  (2022) 008} [\href{https://arxiv.org/abs/2112.12828}{{\ttfamily
  2112.12828}}].

\bibitem{Chandrasekaran:2022cip}
V.~Chandrasekaran, R.~Longo, G.~Penington and E.~Witten, \emph{{An algebra of
  observables for de Sitter space}},
  \href{https://doi.org/10.1007/JHEP02(2023)082}{\emph{JHEP} {\bfseries 02}
  (2023) 082} [\href{https://arxiv.org/abs/2206.10780}{{\ttfamily
  2206.10780}}].

\bibitem{Jensen:2023yxy}
K.~Jensen, J.~Sorce and A.J.~Speranza, \emph{{Generalized entropy for general
  subregions in quantum gravity}},
  \href{https://doi.org/10.1007/JHEP12(2023)020}{\emph{JHEP} {\bfseries 12}
  (2023) 020} [\href{https://arxiv.org/abs/2306.01837}{{\ttfamily
  2306.01837}}].

\bibitem{Kudler-Flam:2023hkl}
J.~Kudler-Flam, S.~Leutheusser, A.A.~Rahman, G.~Satishchandran and
  A.J.~Speranza, \emph{{Covariant regulator for entanglement entropy: Proofs of
  the Bekenstein bound and the quantum null energy condition}},
  \href{https://doi.org/10.1103/PhysRevD.111.105001}{\emph{Phys. Rev. D}
  {\bfseries 111} (2025) 105001}
  [\href{https://arxiv.org/abs/2312.07646}{{\ttfamily 2312.07646}}].

\bibitem{Colafranceschi:2023moh}
E.~Colafranceschi, X.~Dong, D.~Marolf and Z.~Wang, \emph{{Algebras and Hilbert
  spaces from gravitational path integrals. Understanding Ryu-Takayanagi/HRT as
  entropy without AdS/CFT}},
  \href{https://doi.org/10.1007/JHEP10(2024)063}{\emph{JHEP} {\bfseries 10}
  (2024) 063} [\href{https://arxiv.org/abs/2310.02189}{{\ttfamily
  2310.02189}}].

\bibitem{Faulkner:2024gst}
T.~Faulkner and A.J.~Speranza, \emph{{Gravitational algebras and the
  generalized second law}},
  \href{https://doi.org/10.1007/JHEP11(2024)099}{\emph{JHEP} {\bfseries 11}
  (2024) 099} [\href{https://arxiv.org/abs/2405.00847}{{\ttfamily
  2405.00847}}].

\bibitem{Kudler-Flam:2023qfl}
J.~Kudler-Flam, S.~Leutheusser and G.~Satishchandran, \emph{{Generalized black
  hole entropy is von Neumann entropy}},
  \href{https://doi.org/10.1103/PhysRevD.111.025013}{\emph{Phys. Rev. D}
  {\bfseries 111} (2025) 025013}
  [\href{https://arxiv.org/abs/2309.15897}{{\ttfamily 2309.15897}}].

\bibitem{Kudler-Flam:2024psh}
J.~Kudler-Flam, S.~Leutheusser and G.~Satishchandran, \emph{{Algebraic
  Observational Cosmology}},
  \href{https://arxiv.org/abs/2406.01669}{{\ttfamily 2406.01669}}.

\bibitem{AliAhmad:2023etg}
S.~Ali~Ahmad and R.~Jefferson, \emph{{Crossed product algebras and generalized
  entropy for subregions}},
  \href{https://doi.org/10.21468/SciPostPhysCore.7.2.020}{\emph{SciPost Phys.
  Core} {\bfseries 7} (2024) 020}
  [\href{https://arxiv.org/abs/2306.07323}{{\ttfamily 2306.07323}}].

\bibitem{Leutheusser:2021qhd}
S.~Leutheusser and H.~Liu, \emph{{Causal connectability between quantum systems
  and the black hole interior in holographic duality}},
  \href{https://doi.org/10.1103/PhysRevD.108.086019}{\emph{Phys. Rev. D}
  {\bfseries 108} (2023) 086019}
  [\href{https://arxiv.org/abs/2110.05497}{{\ttfamily 2110.05497}}].

\bibitem{Leutheusser:2021frk}
S.A.W.~Leutheusser and H.~Liu, \emph{{Emergent Times in Holographic Duality}},
  \href{https://doi.org/10.1103/PhysRevD.108.086020}{\emph{Phys. Rev. D}
  {\bfseries 108} (2023) 086020}
  [\href{https://arxiv.org/abs/2112.12156}{{\ttfamily 2112.12156}}].

\bibitem{Leutheusser:2022bgi}
S.~Leutheusser and H.~Liu, \emph{{Subregion-subalgebra duality: Emergence of
  space and time in holography}},
  \href{https://doi.org/10.1103/PhysRevD.111.066021}{\emph{Phys. Rev. D}
  {\bfseries 111} (2025) 066021}
  [\href{https://arxiv.org/abs/2212.13266}{{\ttfamily 2212.13266}}].

\bibitem{Engelhardt:2023xer}
N.~Engelhardt and H.~Liu, \emph{{Algebraic ER=EPR and complexity transfer}},
  \href{https://doi.org/10.1007/JHEP07(2024)013}{\emph{JHEP} {\bfseries 07}
  (2024) 013} [\href{https://arxiv.org/abs/2311.04281}{{\ttfamily
  2311.04281}}].

\bibitem{Yngvason:2004uh}
J.~Yngvason, \emph{{The Role of type III factors in quantum field theory}},
  \href{https://doi.org/10.1016/S0034-4877(05)80009-6}{\emph{Rept. Math. Phys.}
  {\bfseries 55} (2005) 135}
  [\href{https://arxiv.org/abs/math-ph/0411058}{{\ttfamily math-ph/0411058}}].

\bibitem{Fredenhagen:1984dc}
K.~Fredenhagen, \emph{{On the Modular Structure of Local Algebras of
  Observables}}, \href{https://doi.org/10.1007/BF01206179}{\emph{Commun. Math.
  Phys.} {\bfseries 97} (1985) 79}.

\bibitem{Araki:1963klf}
H.~Araki, \emph{{A Lattice of Von Neumann Algebras Associated with the Quantum
  Theory of a Free Bose Field}},
  \href{https://doi.org/10.1063/1.1703912}{\emph{J. Math. Phys.} {\bfseries 4}
  (1963) 1343}.

\bibitem{Araki:1964lyc}
H.~Araki, \emph{{Type of von Neumann Algebra Associated with Free Field}},
  \href{https://doi.org/10.1143/ptp.32.956}{\emph{Prog. Theor. Phys.}
  {\bfseries 32} (1964) 956}.

\bibitem{Witten:2018zxz}
E.~Witten, \emph{{APS Medal for Exceptional Achievement in Research: Invited
  article on entanglement properties of quantum field theory}},
  \href{https://doi.org/10.1103/RevModPhys.90.045003}{\emph{Rev. Mod. Phys.}
  {\bfseries 90} (2018) 045003}
  [\href{https://arxiv.org/abs/1803.04993}{{\ttfamily 1803.04993}}].

\bibitem{Chandrasekaran:2022eqq}
V.~Chandrasekaran, G.~Penington and E.~Witten, \emph{{Large N algebras and
  generalized entropy}},
  \href{https://doi.org/10.1007/JHEP04(2023)009}{\emph{JHEP} {\bfseries 04}
  (2023) 009} [\href{https://arxiv.org/abs/2209.10454}{{\ttfamily
  2209.10454}}].

\bibitem{Penington:2023dql}
G.~Penington and E.~Witten, \emph{{Algebras and States in JT Gravity}},
  \href{https://arxiv.org/abs/2301.07257}{{\ttfamily 2301.07257}}.

\bibitem{Kolchmeyer:2023gwa}
D.K.~Kolchmeyer, \emph{{von Neumann algebras in JT gravity}},
  \href{https://doi.org/10.1007/JHEP06(2023)067}{\emph{JHEP} {\bfseries 06}
  (2023) 067} [\href{https://arxiv.org/abs/2303.04701}{{\ttfamily
  2303.04701}}].

\bibitem{Klinger:2026tws}
M.S.~Klinger, J.~Kudler-Flam and G.~Satishchandran, \emph{{Generalized Entropy
  is von Neumann Entropy II: The complete symmetry group and edge modes}},
  \href{https://arxiv.org/abs/2601.07910}{{\ttfamily 2601.07910}}.

\bibitem{Lin:2022rbf}
H.W.~Lin, \emph{{The bulk Hilbert space of double scaled SYK}},
  \href{https://doi.org/10.1007/JHEP11(2022)060}{\emph{JHEP} {\bfseries 11}
  (2022) 060} [\href{https://arxiv.org/abs/2208.07032}{{\ttfamily
  2208.07032}}].

\bibitem{Witten:2023qsv}
E.~Witten, \emph{{Algebras, regions, and observers.}},
  \href{https://doi.org/10.1090/pspum/107/01954}{\emph{Proc. Symp. Pure Math.}
  {\bfseries 107} (2024) 247}
  [\href{https://arxiv.org/abs/2303.02837}{{\ttfamily 2303.02837}}].

\bibitem{DeVuyst:2024khu}
J.~De~Vuyst, S.~Eccles, P.A.~Hoehn and J.~Kirklin, \emph{{Gravitational entropy
  is observer-dependent}},
  \href{https://doi.org/10.1007/JHEP07(2025)146}{\emph{JHEP} {\bfseries 07}
  (2025) 146} [\href{https://arxiv.org/abs/2405.00114}{{\ttfamily
  2405.00114}}].

\bibitem{Sorce:2024zme}
J.~Sorce, \emph{{Analyticity and the Unruh effect: a study of local modular
  flow}}, \href{https://doi.org/10.1007/JHEP09(2024)040}{\emph{JHEP} {\bfseries
  24} (2024) 040} [\href{https://arxiv.org/abs/2403.18937}{{\ttfamily
  2403.18937}}].

\bibitem{Fewster:2024pur}
J.C.~Fewster, D.W.~Janssen, L.D.~Loveridge, K.~Rejzner and J.~Waldron,
  \emph{{Quantum Reference Frames, Measurement Schemes and the Type of Local
  Algebras in Quantum Field Theory}},
  \href{https://doi.org/10.1007/s00220-024-05180-7}{\emph{Commun. Math. Phys.}
  {\bfseries 406} (2025) 19}
  [\href{https://arxiv.org/abs/2403.11973}{{\ttfamily 2403.11973}}].

\bibitem{DeVuyst:2024fxc}
J.~De~Vuyst, S.~Eccles, P.A.~Hoehn and J.~Kirklin, \emph{{Crossed products and
  quantum reference frames: on the observer-dependence of gravitational
  entropy}}, \href{https://doi.org/10.1007/JHEP07(2025)063}{\emph{JHEP}
  {\bfseries 07} (2025) 063}
  [\href{https://arxiv.org/abs/2412.15502}{{\ttfamily 2412.15502}}].

\bibitem{Xu:2024hoc}
J.~Xu, \emph{{Von Neumann algebras in double-scaled SYK}},
  \href{https://doi.org/10.1007/JHEP03(2026)016}{\emph{JHEP} {\bfseries 03}
  (2026) 016} [\href{https://arxiv.org/abs/2403.09021}{{\ttfamily
  2403.09021}}].

\bibitem{AliAhmad:2024wja}
S.~Ali~Ahmad, W.~Chemissany, M.S.~Klinger and R.G.~Leigh, \emph{{Quantum
  reference frames from top-down crossed products}},
  \href{https://doi.org/10.1103/PhysRevD.110.065003}{\emph{Phys. Rev. D}
  {\bfseries 110} (2024) 065003}
  [\href{https://arxiv.org/abs/2405.13884}{{\ttfamily 2405.13884}}].

\bibitem{Bahiru:2022oas}
E.~Bahiru, A.~Belin, K.~Papadodimas, G.~Sarosi and N.~Vardian,
  \emph{{State-dressed local operators in the AdS/CFT correspondence}},
  \href{https://doi.org/10.1103/PhysRevD.108.086035}{\emph{Phys. Rev. D}
  {\bfseries 108} (2023) 086035}
  [\href{https://arxiv.org/abs/2209.06845}{{\ttfamily 2209.06845}}].

\bibitem{Klinger:2023auu}
M.S.~Klinger and R.G.~Leigh, \emph{{Crossed products, conditional expectations
  and constraint quantization}},
  \href{https://doi.org/10.1016/j.nuclphysb.2024.116622}{\emph{Nucl. Phys. B}
  {\bfseries 1006} (2024) 116622}
  [\href{https://arxiv.org/abs/2312.16678}{{\ttfamily 2312.16678}}].

\bibitem{Gomez:2023wrq}
C.~Gomez, \emph{{Entanglement, Observers and Cosmology: a view from von Neumann
  Algebras}},  \href{https://arxiv.org/abs/2302.14747}{{\ttfamily 2302.14747}}.

\bibitem{AliAhmad:2024eun}
S.~Ali~Ahmad, M.S.~Klinger and S.~Lin, \emph{{Semifinite von Neumann algebras
  in gauge theory and gravity}},
  \href{https://doi.org/10.1103/PhysRevD.111.045006}{\emph{Phys. Rev. D}
  {\bfseries 111} (2025) 045006}
  [\href{https://arxiv.org/abs/2407.01695}{{\ttfamily 2407.01695}}].

\bibitem{vanderHeijden:2024tdk}
J.~van~der Heijden and E.~Verlinde, \emph{{An operator algebraic approach to
  black hole information}},
  \href{https://doi.org/10.1007/JHEP02(2025)207}{\emph{JHEP} {\bfseries 02}
  (2025) 207} [\href{https://arxiv.org/abs/2408.00071}{{\ttfamily
  2408.00071}}].

\bibitem{AliAhmad:2024saq}
S.~Ali~Ahmad and M.S.~Klinger, \emph{{Emergent geometry from quantum
  probability}}, \href{https://doi.org/10.1103/PhysRevD.111.105015}{\emph{Phys.
  Rev. D} {\bfseries 111} (2025) 105015}
  [\href{https://arxiv.org/abs/2411.07288}{{\ttfamily 2411.07288}}].

\bibitem{AliAhmad:2025oli}
S.~Ali~Ahmad and M.S.~Klinger, \emph{{Extensions from within}},
  \href{https://arxiv.org/abs/2503.02944}{{\ttfamily 2503.02944}}.

\bibitem{Klinger:2025tvg}
M.~Klinger, \emph{{A Theory of Backgrounds and Background Independence}},
  \href{https://arxiv.org/abs/2512.05043}{{\ttfamily 2512.05043}}.

\bibitem{Klinger:2026kqj}
M.S.~Klinger, \emph{{How to have your wormholes and factorize, too}},
  \href{https://arxiv.org/abs/2602.15120}{{\ttfamily 2602.15120}}.

\bibitem{Aguilar-Gutierrez:2023odp}
S.E.~Aguilar-Gutierrez, E.~Bahiru and R.~Esp{\'\i}ndola, \emph{{The
  centaur-algebra of observables}},
  \href{https://doi.org/10.1007/JHEP03(2024)008}{\emph{JHEP} {\bfseries 03}
  (2024) 008} [\href{https://arxiv.org/abs/2307.04233}{{\ttfamily
  2307.04233}}].

\bibitem{Espindola:2026ekv}
R.~Esp{\'\i}ndola and A.F.~Ali, \emph{{Spectral Admissibility of Real Observers
  in Euclidean de Sitter Gravity}},
  \href{https://arxiv.org/abs/2605.30423}{{\ttfamily 2605.30423}}.

\bibitem{Maldacena:1997re}
J.M.~Maldacena, \emph{{The Large N limit of superconformal field theories and
  supergravity}}, \href{https://doi.org/10.4310/ATMP.1998.v2.n2.a1}{\emph{Adv.
  Theor. Math. Phys.} {\bfseries 2} (1998) 231}
  [\href{https://arxiv.org/abs/hep-th/9711200}{{\ttfamily hep-th/9711200}}].

\bibitem{Witten:1998qj}
E.~Witten, \emph{{Anti-de Sitter space and holography}},
  \href{https://doi.org/10.4310/ATMP.1998.v2.n2.a2}{\emph{Adv. Theor. Math.
  Phys.} {\bfseries 2} (1998) 253}
  [\href{https://arxiv.org/abs/hep-th/9802150}{{\ttfamily hep-th/9802150}}].

\bibitem{Gubser:1998bc}
S.S.~Gubser, I.R.~Klebanov and A.M.~Polyakov, \emph{{Gauge theory correlators
  from noncritical string theory}},
  \href{https://doi.org/10.1016/S0370-2693(98)00377-3}{\emph{Phys. Lett. B}
  {\bfseries 428} (1998) 105}
  [\href{https://arxiv.org/abs/hep-th/9802109}{{\ttfamily hep-th/9802109}}].

\bibitem{Ryu:2006bv}
S.~Ryu and T.~Takayanagi, \emph{{Holographic derivation of entanglement entropy
  from AdS/CFT}},
  \href{https://doi.org/10.1103/PhysRevLett.96.181602}{\emph{Phys. Rev. Lett.}
  {\bfseries 96} (2006) 181602}
  [\href{https://arxiv.org/abs/hep-th/0603001}{{\ttfamily hep-th/0603001}}].

\bibitem{Hubeny:2007xt}
V.E.~Hubeny, M.~Rangamani and T.~Takayanagi, \emph{{A Covariant holographic
  entanglement entropy proposal}},
  \href{https://doi.org/10.1088/1126-6708/2007/07/062}{\emph{JHEP} {\bfseries
  07} (2007) 062} [\href{https://arxiv.org/abs/0705.0016}{{\ttfamily
  0705.0016}}].

\bibitem{Casini_2011}
H.~Casini, M.~Huerta and R.C.~Myers, \emph{Towards a derivation of holographic
  entanglement entropy},
  \href{https://doi.org/10.1007/jhep05(2011)036}{\emph{Journal of High Energy
  Physics} {\bfseries 2011} (2011) }.

\bibitem{Lewkowycz:2013nqa}
A.~Lewkowycz and J.~Maldacena, \emph{{Generalized gravitational entropy}},
  \href{https://doi.org/10.1007/JHEP08(2013)090}{\emph{JHEP} {\bfseries 08}
  (2013) 090} [\href{https://arxiv.org/abs/1304.4926}{{\ttfamily 1304.4926}}].

\bibitem{Nishioka:2018khk}
T.~Nishioka, \emph{{Entanglement entropy: holography and renormalization
  group}}, \href{https://doi.org/10.1103/RevModPhys.90.035007}{\emph{Rev. Mod.
  Phys.} {\bfseries 90} (2018) 035007}
  [\href{https://arxiv.org/abs/1801.10352}{{\ttfamily 1801.10352}}].

\bibitem{Faulkner:2013ana}
T.~Faulkner, A.~Lewkowycz and J.~Maldacena, \emph{{Quantum corrections to
  holographic entanglement entropy}},
  \href{https://doi.org/10.1007/JHEP11(2013)074}{\emph{JHEP} {\bfseries 11}
  (2013) 074} [\href{https://arxiv.org/abs/1307.2892}{{\ttfamily 1307.2892}}].

\bibitem{Engelhardt:2014gca}
N.~Engelhardt and A.C.~Wall, \emph{{Quantum Extremal Surfaces: Holographic
  Entanglement Entropy beyond the Classical Regime}},
  \href{https://doi.org/10.1007/JHEP01(2015)073}{\emph{JHEP} {\bfseries 01}
  (2015) 073} [\href{https://arxiv.org/abs/1408.3203}{{\ttfamily 1408.3203}}].

\bibitem{Almheiri_2015}
A.~Almheiri, X.~Dong and D.~Harlow, \emph{Bulk locality and quantum error
  correction in {AdS}/{CFT}},
  \href{https://doi.org/10.1007/jhep04(2015)163}{\emph{Journal of High Energy
  Physics} {\bfseries 2015} (2015) }.

\bibitem{Dong:2016eik}
X.~Dong, D.~Harlow and A.C.~Wall, \emph{{Reconstruction of Bulk Operators
  within the Entanglement Wedge in Gauge-Gravity Duality}},
  \href{https://doi.org/10.1103/PhysRevLett.117.021601}{\emph{Phys. Rev. Lett.}
  {\bfseries 117} (2016) 021601}
  [\href{https://arxiv.org/abs/1601.05416}{{\ttfamily 1601.05416}}].

\bibitem{Harlow:2016vwg}
D.~Harlow, \emph{{The Ryu\textendash{}Takayanagi Formula from Quantum Error
  Correction}}, \href{https://doi.org/10.1007/s00220-017-2904-z}{\emph{Commun.
  Math. Phys.} {\bfseries 354} (2017) 865}
  [\href{https://arxiv.org/abs/1607.03901}{{\ttfamily 1607.03901}}].

\bibitem{Zhong:2024fmn}
H.~Zhong, \emph{{Probing the Page transition via approximate quantum error
  correction}}, \href{https://doi.org/10.1007/JHEP01(2025)086}{\emph{JHEP}
  {\bfseries 01} (2025) 086}
  [\href{https://arxiv.org/abs/2408.15104}{{\ttfamily 2408.15104}}].

\bibitem{Faulkner:2020hzi}
T.~Faulkner, \emph{{The holographic map as a conditional expectation}},
  \href{https://arxiv.org/abs/2008.04810}{{\ttfamily 2008.04810}}.

\bibitem{Akers:2021fut}
C.~Akers and G.~Penington, \emph{{Quantum minimal surfaces from quantum error
  correction}},
  \href{https://doi.org/10.21468/SciPostPhys.12.5.157}{\emph{SciPost Phys.}
  {\bfseries 12} (2022) 157}
  [\href{https://arxiv.org/abs/2109.14618}{{\ttfamily 2109.14618}}].

\bibitem{Gesteau:2023hbq}
E.~Gesteau, \emph{{Large N von Neumann Algebras and the Renormalization of
  Newton{\textquoteright}s Constant}},
  \href{https://doi.org/10.1007/s00220-024-05192-3}{\emph{Commun. Math. Phys.}
  {\bfseries 406} (2025) 40}
  [\href{https://arxiv.org/abs/2302.01938}{{\ttfamily 2302.01938}}].

\bibitem{Czech:2012bh}
B.~Czech, J.L.~Karczmarek, F.~Nogueira and M.~Van~Raamsdonk, \emph{{The Gravity
  Dual of a Density Matrix}},
  \href{https://doi.org/10.1088/0264-9381/29/15/155009}{\emph{Class. Quant.
  Grav.} {\bfseries 29} (2012) 155009}
  [\href{https://arxiv.org/abs/1204.1330}{{\ttfamily 1204.1330}}].

\bibitem{Hamilton:2006az}
A.~Hamilton, D.N.~Kabat, G.~Lifschytz and D.A.~Lowe, \emph{{Holographic
  representation of local bulk operators}},
  \href{https://doi.org/10.1103/PhysRevD.74.066009}{\emph{Phys. Rev. D}
  {\bfseries 74} (2006) 066009}
  [\href{https://arxiv.org/abs/hep-th/0606141}{{\ttfamily hep-th/0606141}}].

\bibitem{Morrison:2014jha}
I.A.~Morrison, \emph{{Boundary-to-bulk maps for AdS causal wedges and the
  Reeh-Schlieder property in holography}},
  \href{https://doi.org/10.1007/JHEP05(2014)053}{\emph{JHEP} {\bfseries 05}
  (2014) 053} [\href{https://arxiv.org/abs/1403.3426}{{\ttfamily 1403.3426}}].

\bibitem{Bousso:2012sj}
R.~Bousso, S.~Leichenauer and V.~Rosenhaus, \emph{{Light-sheets and AdS/CFT}},
  \href{https://doi.org/10.1103/PhysRevD.86.046009}{\emph{Phys. Rev. D}
  {\bfseries 86} (2012) 046009}
  [\href{https://arxiv.org/abs/1203.6619}{{\ttfamily 1203.6619}}].

\bibitem{Bousso:2012mh}
R.~Bousso, B.~Freivogel, S.~Leichenauer, V.~Rosenhaus and C.~Zukowski,
  \emph{{Null Geodesics, Local CFT Operators and AdS/CFT for Subregions}},
  \href{https://doi.org/10.1103/PhysRevD.88.064057}{\emph{Phys. Rev. D}
  {\bfseries 88} (2013) 064057}
  [\href{https://arxiv.org/abs/1209.4641}{{\ttfamily 1209.4641}}].

\bibitem{Hubeny:2012wa}
V.E.~Hubeny and M.~Rangamani, \emph{{Causal Holographic Information}},
  \href{https://doi.org/10.1007/JHEP06(2012)114}{\emph{JHEP} {\bfseries 06}
  (2012) 114} [\href{https://arxiv.org/abs/1204.1698}{{\ttfamily 1204.1698}}].

\bibitem{Wall:2012uf}
A.C.~Wall, \emph{{Maximin Surfaces, and the Strong Subadditivity of the
  Covariant Holographic Entanglement Entropy}},
  \href{https://doi.org/10.1088/0264-9381/31/22/225007}{\emph{Class. Quant.
  Grav.} {\bfseries 31} (2014) 225007}
  [\href{https://arxiv.org/abs/1211.3494}{{\ttfamily 1211.3494}}].

\bibitem{Headrick:2014cta}
M.~Headrick, V.E.~Hubeny, A.~Lawrence and M.~Rangamani, \emph{{Causality \&
  holographic entanglement entropy}},
  \href{https://doi.org/10.1007/JHEP12(2014)162}{\emph{JHEP} {\bfseries 12}
  (2014) 162} [\href{https://arxiv.org/abs/1408.6300}{{\ttfamily 1408.6300}}].

\bibitem{Jafferis:2015del}
D.L.~Jafferis, A.~Lewkowycz, J.~Maldacena and S.J.~Suh, \emph{{Relative entropy
  equals bulk relative entropy}},
  \href{https://doi.org/10.1007/JHEP06(2016)004}{\emph{JHEP} {\bfseries 06}
  (2016) 004} [\href{https://arxiv.org/abs/1512.06431}{{\ttfamily
  1512.06431}}].

\bibitem{Polchinski:1999yd}
J.~Polchinski, L.~Susskind and N.~Toumbas, \emph{{Negative energy,
  superluminosity and holography}},
  \href{https://doi.org/10.1103/PhysRevD.60.084006}{\emph{Phys. Rev. D}
  {\bfseries 60} (1999) 084006}
  [\href{https://arxiv.org/abs/hep-th/9903228}{{\ttfamily hep-th/9903228}}].

\bibitem{Harlow:2018fse}
D.~Harlow, \emph{{TASI Lectures on the Emergence of Bulk Physics in AdS/CFT}},
  \href{https://doi.org/10.22323/1.305.0002}{\emph{PoS} {\bfseries TASI2017}
  (2018) 002} [\href{https://arxiv.org/abs/1802.01040}{{\ttfamily
  1802.01040}}].

\bibitem{Kamal:2019skn}
H.~Kamal and G.~Penington, \emph{{The Ryu-Takayanagi Formula from Quantum Error
  Correction: An Algebraic Treatment of the Boundary CFT}},
  \href{https://arxiv.org/abs/1912.02240}{{\ttfamily 1912.02240}}.

\bibitem{Kang:2018xqy}
M.J.~Kang and D.K.~Kolchmeyer, \emph{{Holographic Relative Entropy in
  Infinite-Dimensional Hilbert Spaces}},
  \href{https://doi.org/10.1007/s00220-022-04627-z}{\emph{Commun. Math. Phys.}
  {\bfseries 400} (2023) 1665}
  [\href{https://arxiv.org/abs/1811.05482}{{\ttfamily 1811.05482}}].

\bibitem{Kang:2019dfi}
M.J.~Kang and D.K.~Kolchmeyer, \emph{{Entanglement wedge reconstruction of
  infinite-dimensional von Neumann algebras using tensor networks}},
  \href{https://doi.org/10.1103/PhysRevD.103.126018}{\emph{Phys. Rev. D}
  {\bfseries 103} (2021) 126018}
  [\href{https://arxiv.org/abs/1910.06328}{{\ttfamily 1910.06328}}].

\bibitem{Xu:2024xbz}
M.~Xu and H.~Zhong, \emph{{Adding the algebraic Ryu-Takayanagi formula to the
  algebraic reconstruction theorem}},
  \href{https://arxiv.org/abs/2411.06361}{{\ttfamily 2411.06361}}.

\bibitem{Crann:2024gkv}
J.~Crann and M.J.~Kang, \emph{{Algebraic approach to spacetime bulk
  reconstruction}},  \href{https://arxiv.org/abs/2412.00298}{{\ttfamily
  2412.00298}}.

\bibitem{Lashkari:2013koa}
N.~Lashkari, M.B.~McDermott and M.~Van~Raamsdonk, \emph{{Gravitational dynamics
  from entanglement 'thermodynamics'}},
  \href{https://doi.org/10.1007/JHEP04(2014)195}{\emph{JHEP} {\bfseries 04}
  (2014) 195} [\href{https://arxiv.org/abs/1308.3716}{{\ttfamily 1308.3716}}].

\bibitem{Faulkner:2013ica}
T.~Faulkner, M.~Guica, T.~Hartman, R.C.~Myers and M.~Van~Raamsdonk,
  \emph{{Gravitation from Entanglement in Holographic CFTs}},
  \href{https://doi.org/10.1007/JHEP03(2014)051}{\emph{JHEP} {\bfseries 03}
  (2014) 051} [\href{https://arxiv.org/abs/1312.7856}{{\ttfamily 1312.7856}}].

\bibitem{Swingle:2014uza}
B.~Swingle and M.~Van~Raamsdonk, \emph{{Universality of Gravity from
  Entanglement}},  \href{https://arxiv.org/abs/1405.2933}{{\ttfamily
  1405.2933}}.

\bibitem{Faulkner:2017tkh}
T.~Faulkner, F.M.~Haehl, E.~Hijano, O.~Parrikar, C.~Rabideau and
  M.~Van~Raamsdonk, \emph{{Nonlinear Gravity from Entanglement in Conformal
  Field Theories}}, \href{https://doi.org/10.1007/JHEP08(2017)057}{\emph{JHEP}
  {\bfseries 08} (2017) 057}
  [\href{https://arxiv.org/abs/1705.03026}{{\ttfamily 1705.03026}}].

\bibitem{Lewkowycz:2018sgn}
A.~Lewkowycz and O.~Parrikar, \emph{{The holographic shape of entanglement and
  Einstein{\textquoteright}s equations}},
  \href{https://doi.org/10.1007/JHEP05(2018)147}{\emph{JHEP} {\bfseries 05}
  (2018) 147} [\href{https://arxiv.org/abs/1802.10103}{{\ttfamily
  1802.10103}}].

\bibitem{Vaughan2009}
V.F.~Jones, ``Von neumann algebras.'' https://math.berkeley.edu/~vfr/MATH20909/
  VonNeumann2009.pdf, 2009.

\bibitem{Sorce:2023fdx}
J.~Sorce, \emph{{Notes on the type classification of von Neumann algebras}},
  \href{https://doi.org/10.1142/S0129055X24300024}{\emph{Rev. Math. Phys.}
  {\bfseries 36} (2024) 2430002}
  [\href{https://arxiv.org/abs/2302.01958}{{\ttfamily 2302.01958}}].

\bibitem{reed1972methods}
M.~Reed, B.~Simon, B.~Simon and B.~Simon, \emph{Methods of modern mathematical
  physics}, vol.~1, Academic press New York (1972).

\bibitem{takesaki2006tomita}
M.~Takesaki, \emph{Tomita's theory of modular Hilbert algebras and its
  applications}, vol.~128, Springer (2006).

\bibitem{Liu:2025krl}
H.~Liu, \emph{{Lectures on entanglement, von Neumann algebras, and emergence of
  spacetime}},  in \emph{{Theoretical Advanced Study Institute in Elementary
  Particle Physics 2023}: {Aspects of Symmetry}}, 10, 2025
  [\href{https://arxiv.org/abs/2510.07017}{{\ttfamily 2510.07017}}].

\bibitem{Zhong:2026syt}
H.~Zhong, \emph{{An algebraic description of the Page transition}},
  \href{https://doi.org/10.1007/JHEP04(2026)160}{\emph{JHEP} {\bfseries 04}
  (2026) 160} [\href{https://arxiv.org/abs/2601.11363}{{\ttfamily
  2601.11363}}].

\bibitem{araki1975relative}
H.~Araki, \emph{Relative entropy of states of von neumann algebras},
  {\emph{Publications of the Research Institute for Mathematical Sciences}
  {\bfseries 11} (1975) 809}.

\bibitem{araki1975inequalities}
H.~Araki, \emph{Inequalities in von neumann algebras}, {\emph{Les rencontres
  physiciens-math{\'e}maticiens de Strasbourg-RCP25} {\bfseries 22} (1975) 1}.

\bibitem{segal1960note}
I.E.~Segal, \emph{A note on the concept of entropy}, {\emph{Journal of
  Mathematics and Mechanics} (1960) 623}.

\bibitem{ohya2004quantum}
M.~Ohya and D.~Petz, \emph{Quantum entropy and its use}, Springer Science \&
  Business Media (2004).

\bibitem{Longo:2022lod}
R.~Longo and E.~Witten, \emph{{A note on continuous entropy}},
  \href{https://doi.org/10.4310/PAMQ.2023.v19.n5.a5}{\emph{Pure Appl. Math.
  Quart.} {\bfseries 19} (2023) 2501}
  [\href{https://arxiv.org/abs/2202.03357}{{\ttfamily 2202.03357}}].

\bibitem{takesaki1973duality}
M.~Takesaki, \emph{Duality for crossed products and the structure of von
  neumann algebras of type iii}, .

\bibitem{connes1973classification}
A.~Connes, \emph{Une classification des facteurs de type iii},  in
  \emph{Annales scientifiques de l'{\'E}cole Normale Sup{\'e}rieure}, vol.~6,
  pp.~133--252, 1973.

\bibitem{connes2008noncommutative}
A.~Connes, J.~Cuntz and M.A.~Rieffel, \emph{Noncommutative geometry},
  {\emph{Oberwolfach Reports} {\bfseries 4} (2008) 2543}.

\end{thebibliography}\endgroup
\end{document}